\documentclass{article}

\usepackage{url}		

\usepackage{amssymb}
\usepackage{amsthm}
\usepackage{amsmath}

\usepackage{algorithm2e}

\usepackage{float}

\usepackage{mathtools}

\usepackage{stmaryrd}

\usepackage{amsmath,pict2e}

\makeatletter
\newcommand{\bigcomp}{%
  \DOTSB
  \mathop{\vphantom{\sum}\mathpalette\bigcomp@\relax}%
  \slimits@
}
\newcommand{\bigcomp@}[2]{%
  \begingroup\m@th
  \sbox\z@{$#1\sum$}%
  \setlength{\unitlength}{0.9\dimexpr\ht\z@+\dp\z@}%
  \vcenter{\hbox{%
    \begin{picture}(1,1)
    \bigcomp@linethickness{#1}
    \put(0.5,0.5){\circle{1}}
    \end{picture}%
  }}%
  \endgroup
}
\newcommand{\bigcomp@linethickness}[1]{%
  \linethickness{%
      \ifx#1\displaystyle 2\fontdimen8\textfont\else
      \ifx#1\textstyle 1.65\fontdimen8\textfont\else
      \ifx#1\scriptstyle 1.65\fontdimen8\scriptfont\else
      1.65\fontdimen8\scriptscriptfont\fi\fi\fi 3
  }%
}

\makeatother

\usepackage{gensymb}

\usepackage{tikz}
\usetikzlibrary{automata,positioning}

\usepackage{multicol}

\usepackage{graphicx}
\usepackage{dsfont}

\def\Z{\mathbb Z} 
\def\N{\mathbb N} 
\def\R{\mathbb R} 

\def\01{0\!\!\!1}

\def\dy{\displaystyle}

\def\Flid{\text{\rm Flid}}

\def\Init{\text{\rm Init}}
\def\Last{\text{\rm Last}}

\def\suc{\text{\rm succ}}
\def\pre{\text{\rm pred}}

\def\Godot{\underline{\odot}}
\def\Fact{\text{\rm Fact}}
\def\mA{\mathcal{A}}
\def\mB{\mathcal{B}}

\def\mF{\mathcal{F}}
\def\mG{\mathcal{G}}
\def\mH{\mathcal{H}}
\def\mI{\mathcal{I}}

\def\mK{\mathcal{K}}

\def\mM{\mathcal{M}}

\def\mP{\mathcal{P}}

\def\mV{\mathcal{V}}
\def\mW{\mathcal{W}}
\def\mX{\mathcal{X}}
\def\mY{\mathcal{Y}}
\def\mZ{\mathcal{Z}}

\usepackage{multirow,bigdelim}

\theoremstyle{plain}
\newtheorem{theorem}{Theorem}[subsection]
\newtheorem{proposition}{Proposition}[subsection]
\newtheorem{corollary}{Corollary}[subsection]

\theoremstyle{definition}
\newtheorem{definition}{Definition}[subsection]

\theoremstyle{remark}
\newtheorem*{remark}{Remark}

\usepackage{arydshln}
\usepackage{xcolor}

\usepackage[margin=1cm]{caption}
\usepackage[font=scriptsize]{caption} 

\title{Reflective Numeration Systems I: a Global Standpoint}

\author{Beno\^{\i}t Rittaud\footnote{Universit\'e Sorbonne Paris Nord, LAGA, CNRS, UMR 7539, F-93430 Villetaneuse, France.
\texttt{rittaud@math.univ-paris13.fr}}}

\begin{document}

\maketitle

\begin{abstract} We present a framework to generalize the standard $b$-ary Gray code to get the $k$-bonacci ones obtained in \cite{Bernini2013} as well as many others by using theoretical tools that allow to make calculations on lists. We introduce the notion of {\em $\mZ$-Gray product}, from which we deduce sequences of lists of finite words avoiding a predefinite list $\mZ$ of factors and which satisfy a power-associativity property as well a generalizations of the classical flipping digit property.
\end{abstract}



\section{Introduction}\label{Introduction}

\subsection{Informal presentations of standard $b$-ary and $k$-bonacci Gray codes}\label{IntroGrayCodesConnus}

In 1953, Franck Gray \cite{Gray} presented the first algorithm designed to produce the infinite list of all finite binary words, each appearing exactly once, in such an order that each element differs from its successor by a single digit. Such a list of words is said to {\em satisfy the flipping digit property} (or, when the notion is extended later on in this paper, the {\em $\{1\}$-flipping digit property}). Historically, the concept appeared indirectly more half of a century before in the {\em baguenaudier}, a puzzle investigated by \'Edouard Lucas in the context of recreational mathematics \cite{Hinz}.

Gray's list can be presented inductively as follows: start with the list made of the single empty word, then, the $2^{\ell}$ first elements of the list being written for some $\ell\in\N$, write them a second time in reversed order, left-concatenate a $1$ to each of them, then left-concatenate a $0$ to each element of the first half (Table \ref{GrayTable}).
When $\ell$ goes to infinity, this construction leads to an infinite list $\mG$ that contains all finite binary words, each exactly once (assuming the equivalence between a word and the same word with a $0$ left-appended to it). Hence, the natural bijection that maps $\N$ onto $\mG$ is an alternative one-to-one correspondence between natural integers and finite binary words, and from Gray's construction it is easy to prove that $\mG$ satisfies the flipping digit property.

\begin{table}[h]
\centering
\begin{tabular}{cccc}
0&0&0&0\\
\cdashline{4-4}[2pt/2pt]
0&0&0&1\\
\cdashline{3-3}[2pt/2pt]\cline{4-4}
0&0&1&1\\
0&0&1&0\\
\cdashline{2-2}[2pt/2pt]\cline{3-4}
0&1&1&0\\
0&1&1&1\\
0&1&0&1\\
0&1&0&0\\
\cdashline{2-2}[2pt/2pt]\cline{3-4}
1&1&0&0\\
1&1&0&1\\
1&1&1&1\\
1&1&1&0\\
1&0&1&0\\
1&0&1&1\\
1&0&0&1\\
1&0&0&0\\
\end{tabular}
\caption{The first iterations for the standard Gray code.}
\label{GrayTable}
\end{table}

This {\em standard binary Gray code} is now regarded as a particular case of a more general combinatorial framework: a metric denumerable set being given, a  {\em Gray code} is an ordering of this set such that the distance between any two consecutive elements is ``small'' in some predefined sense, possibly with additional constraints on the ordering. To stick to finite binary words, let us mention the one of defining a Gray code such that the flipping digit index is as evenly distributed as possible (whereas, in the standard Gray code, the rightmost digits change more often than the others), with a first result obtained by Tibor Bakos \cite{Adam} and a general solution given by Donald Knuth in \cite[Sec 7.2.1.1]{Knuth}. The reader will find in the survey by Torsten M\"utze \cite{Mutze} a lot of other possible constraints investigated by various authors: {\em weight-monotonicity} (forcing the number of $1$s in the successive words to be essentially increasing), {\em $t$-antipodal} constraint (forcing the distance between a binary string and its complement to be exactly $t$), or generalizations to strings on more general alphabets.

In the present paper, the denomination of {\em Gray code} will concern only lists of finite words on a given finite alphabet, equipped with the Hamming distance.

The first generalization of the standard binary Gray code was introduced by Flores \cite{Flores} as the reflected number systems, here referred as the {\em standard $b$-ary Gray code}. It extends Gray's original idea in a very natural way to words on the alphabet $\llbracket 0,b\llbracket$, where $b\geqslant 2$ is a given integer. It starts with an initial list containing only the empty word $\varepsilon$, then iterates the following process: the current list being given, it is written a second time in reverse order, then a third time in the initial order, etc., alternating between increasing and initial order until $b$ copies are written down. Then, to the words of each copy of the list is left-concatenated a letter, the same for each copy, from $0$ to $b-1$ (Table \ref{BaryTable}).

\begin{table}[h]
\centering
\begin{tabular}{ccc}
0&0&0\\
\cdashline{3-3}[2pt/2pt]
0&0&1\\
0&0&2\\
\cdashline{2-2}[2pt/2pt]\cline{3-3}
0&1&2\\
0&1&1\\
0&1&0\\
0&2&0\\
0&2&1\\
0&2&2\\
\cdashline{1-1}[2pt/2pt]\cline{2-3}
1&2&2\\
1&2&1\\
1&2&0\\
1&1&0\\
1&1&1\\
1&1&2\\
1&0&2\\
1&0&1\\
1&0&0\\
2&0&0\\
2&0&1\\
2&0&2\\
2&1&2\\
2&1&1\\
2&1&0\\
2&2&0\\
2&2&1\\
2&2&2
\end{tabular}
\caption{The first iterations for the standard ternary Gray code.}
\label{BaryTable}
\end{table}

Further generalizations were made to provide lists of words still satisfying the flipping digit property but containing only words avoiding some given factor. The idea is to replicate in a reflection context the properties of the list of codage of numbers given by numeration systems like the Zeckendorf one, or more general ones defined by linear recurring sequences of positive integers. Bernini {\em et al.} \cite{Bernini2013} presented such a Gray code that avoid the factor $100$ (Table \ref{FiboTable}), in which the mirroring process does not apply to the full list but only to the sublist following the penultimate mirror. 

\begin{table}[h]
\begin{center}
\begin{tabular}{ccccc}
 0 & 0 & 0 & 0 & 0 \\
\cdashline{5-5}[2pt/2pt]
 0 & 0 & 0 & 0 & 1 \\
\cdashline{4-4}[2pt/2pt]\cline{5-5}
 0 & 0 & 0 & 1 & 1 \\
 0 & 0 & 0 & 1 & 0 \\
\cdashline{3-3}[2pt/2pt]\cline{4-5}
 0 & 0 & 1 & 1 & 0 \\
 0 & 0 & 1 & 1 & 1 \\
 0 & 0 & 1 & 0 & 1 \\
\cdashline{2-2}[2pt/2pt]\cline{3-5}
 0 & 1 & 1 & 0 & 1 \\
 0 & 1 & 1 & 1 & 1 \\
 0 & 1 & 1 & 1 & 0 \\
 0 & 1 & 0 & 1 & 0 \\
 0 & 1 & 0 & 1 & 1\\
\cdashline{1-1}[2pt/2pt]\cline{2-5}
 1&1&0&1&1\\
 1&1&0&1&0\\
 1&1&1&1&0\\
 1&1&1&1&1\\
 1&1&1&0&1\\
 1&0&1&0&1\\
 1&0&1&1&1\\
 1&0&1&1&0
\end{tabular}
\end{center}
\caption{The beginning of the Fibonacci Gray code of Bernini {\em et al.} \cite{Bernini2013}. Each new mirroring only applies to the lower sublist limited by the penultimate mirror.}
\label{FiboTable}
\end{table}

Other results were obtained by various authors, going back at least to Squire \cite{Squire}, who investigated the possibility of providing a Gray code for lists of words on $\llbracket 0,b\llbracket$ of length exactly $n$ and avoiding some given factor, showing that the possibility of such a Gray code depends on the parity of $b$ and is linked to the autorrelation of the given forbidden factor. More recent approachs are those of Vajnovszki \cite{Vajnovszki}, and Bernini {\em et al.} \cite{Bernini2013,Bernini2015}. The last one focuses on the notion of forbideen factor {\em inducing zero periodicity}, proving that this property is a sufficient hypothesis for such a forbidden factor to allow the existence of a Gray code. In Barcucci {\em et al.} \cite{Barcucci2113,BarcucciWords2021,Barcucci2022} are investigated Gray codes obtained by the factors defined by some linear recurring sequences.

\subsection{Lists building}\label{IntroOperations}

Since the standard $b$-ary and $k$-bonacci cases, in the usual numeration system as well as in the Gray code contexts, are important cases for the settings of the present paper, here is a short presentation of them that complements the previous explanations, with some additional notations. All of this will be made more complete and precise in Section \ref{Basic}.

For any integer $b\geqslant 2$, the {\em standard $b$-ary numeration system} can be seen as the set $\mA^*$ of finite words on the alphabet $\mA\coloneqq\llbracket 0,b\llbracket$, made a list by writing its elements in increasing order for the radix order. To obtain it in a recursive way, consider the sequence of lists of words $(\mB_\ell)_{\ell\in\N}$ defined by $\mB_0=\{\varepsilon\}$ (where $\varepsilon$ is the empty word) and, for any $\ell\in\N^*$, $\mB_\ell$ is the list increasing for the radix order and made of all possible words obtained by left-appending a letter of $\mA$ to a word of $\mB_\ell$. Later on (Section \ref{Lists}), this will be written $\mB_{\ell+1}\coloneqq\sum_{\alpha\in\mA}\alpha\mB_{\ell}$ (the sum being noncommutative here, corresponding to a union of lists made by increasing values of $\alpha$), or, in a factorized way, as $\mB_{\ell+1}=\mA\boldsymbol{\cdot}\mB_\ell$ (see Section \ref{ProductLists}). The sequence $(\mB_\ell)_{\ell\in\N}$ is an increasing sequence of lists, meaning that the list $\mB_\ell$ is the beginning of $\mB_{\ell+1}$ for any $\ell\in\N$, where words like $w$ and $0w$ are identified. Hence, this sequence converges in a natural sense to an infinite list $\mB=(b_n)_{n\in\N}$. Eventually, it is easy to prove that, for any $n\in\N$, $b_n$ is the usual $b$-ary expansion of $n$. Note that this approach to define the $b$-ary numeration system for integers can be regarded as {\em global}: lists of words are built up by some process from which the flipping digit property follows immediately as well as the fact that all desired words are obtained. The more classical {\em local} approach uses numeration systems instead of lists, designing an algorithm to map each integer to a suitable word, thus building the list of words element by element. This local point of view, which is also worth considering in the context of Gray codes, is developped in a separate article \cite{Rittaud2}, made independent of the present one which deals only with the global point of view.

A similar reasoning can be applied to get the standard $b$-ary Gray code. A finite list $\mX$ being given, write $\overleftarrow{\mX}$ for the same list written in reverse order. Also, write $\mX^{\leftrightarrow i}$ for the list $\mX$ to which this reversing operator is applied $i$ times (hence $\mX^{\leftrightarrow 2j}=\mX$ and $\mX^{\leftrightarrow 2j+1}=\overleftarrow{\mX}$). The standard $b$-ary Gray code can then be regarded as the limit sequence $\mG$ of lists $\mG_\ell$ defined inductively as $\mG_0=\{\varepsilon\}$ and $\mG_{\ell+1}=\sum_{\alpha\in\mA}\alpha{\mG_\ell}^{\leftrightarrow \alpha}$. The fact that $\mG$ satisfies the flipping digit property is an easy consequence of this formula, which will be factorized as $\mG_{\ell+1}=\mA\odot\mG_\ell$ in Section \ref{sec:GrayProd}, where $\odot$ will stand for the {\em Gray product}.

Recall that the Zeckendorf-Fibonacci numeration system \cite{Zeckendorf} consists in coding any $n\in\N$ by a binary sequence $(w_i)_{i\in\N}$ (up to the equivalence relation $\sim$ such that $w\sim 0w$) such that $n=\sum_{i\in\N}\omega_iz_i$, where $z_0=1$, $z_1=2$ and $z_i=z_{i-1}+z_{i-2}$ for any $i\geqslant 2$. (Here, the sequence $(z_i)_{i\in\N}$ will be called the {\em Zeckendorf sequence}; it is of course a simple shift of the standard Fibonacci sequence.) The possible sequences $(\omega_i)_{i\in\N}$ satisfying the previous equality are the {\em legal representations} of $n$. For any $\mZ\subset\mA^*$, write $\Fact_\mA(\mZ)$ for the set of elements of $\mA^*_\sim$ containing some element of $\mZ$ as a factor. 
The biggest legal representation for the radix order is the {\em greedy} one, it is characterized by not containing any element of $\Fact_{\{0,1\}}(\{011\})$. Similarly, the {\em lazy} expansion is the smallest one for the radix order, and it is characterized by not containing any element of $\Fact_{\{0,1\}}(\{100\})$. Conversely, any finite binary sequence out of $\Fact_{\{0,1\}}(\{011\})$ (resp. $\Fact_{\{0,1\}}(\{100\})$) is the greedy (resp. lazy) representation of some integer.

Consider the list $\mF\coloneqq(f_n)_{n\in\N}$ such that $f_n$ is the greedy (resp. lazy) Zeckendorf representation of $n$. It is therefore the list of greedy (resp. lazy) words written in increasing radix order. In the greedy case, the list $\mF$ is the limit of the sequence $(\mF_\ell)_{\ell\in\N}$ where $\mF_0=\{\varepsilon\}$, $\mF_1=\{1\}$ and $\mF_\ell=0\mF_{\ell-1}+10\mF_{\ell-2}$ for $\ell\geqslant 2$. As it happens, the lazy case is more significant in the framework of Gray codes (as will be more apparent in the forthcoming paper \cite{Rittaud2}). To get the list $\mF$ of lazy Zeckendorf expansion of integers, we can first write $\mF=\sum_{\ell\in\N}\mM_\ell$, where $(\mM_\ell)_{\ell\in\N}$ is defined by:
\[\mM_0\coloneqq\{\varepsilon\},\ \mM_1\coloneqq\{1\}\text{ and }\mM_\ell\coloneqq 10\mM_{\ell-2}+1\mM_{\ell-1}\text{ for $\ell\geqslant 2$.}\]


A second way, closer to the ideas of the present paper (see Section \ref{sec:ZProduct}), is to obtain $\mF$ as the limit of the sequence $(\mF_\ell)_{\ell\in\N}$ of lists defined by $\mF_0\coloneqq\{\varepsilon\}$ and, for any $\ell\in\N^*$, $\mF_{\ell}\coloneqq{\underline{\boldsymbol{\cdot}}}(0\mF_{\ell-1}+1\mF_{\ell-1})$, where the ${\underline{\boldsymbol{\cdot}}}$ operator, applied to a list, removes all its elements in $\Fact_{\{0,1\}}(\{100\})$.

The {\em Fibonacci Gray code} can be defined in a similar way (see \cite{Bernini2013}), either by defining  $\mF$ as $\sum_{\ell\in\N}\mM_\ell$ with
\[\mathcal{M}_{0}\coloneqq\{\varepsilon\},\ \mathcal{M}_1\coloneqq\{1\} \text{ and } \mathcal{M}_{\ell}\coloneqq 1\overleftarrow{{\mathcal{M}_{\ell-1}}}+10\overleftarrow{\mathcal{M}_{\ell-2}}\ \text{for $\ell\geqslant 2$,}\]
\noindent or, as will be proved in Section \ref{sec:ZProduct} (Proposition \ref{ZFibo}) by $\mF=\lim_\ell(\mF_\ell)$, with $\mF_0\coloneqq\{\varepsilon\}$ and $\mF_\ell\coloneqq{\underline{\boldsymbol{\cdot}}}\left(0\mF_{\ell-1}+1\overleftarrow{\mF_{\ell-1}}\right)$ for $\ell\in\N^*$. As for the standard $
b$-ary Gray code, the flipping digit property comes as a reasonably easy consequence of any of these definitions.

For an integer $k\geqslant 2$, the {\em lazy $k$-bonacci numeration system} is defined similarly as the Zeckendorf-Fibonacci one (the latter corresponding to $k=2$), the sequence $(z_i)_{i\in\N}$ being now defined by $z_i=2^i$ for $i<k$ and $z_i=\sum_{i-k\leqslant j<i}z_{j}$ for $i\geqslant k$ (or, in a more concise way, putting $z_{\text{-}1}\coloneqq 1$, by $\displaystyle z_i=\sum_{\max(\text{-}1,i-k)\leqslant j<i}z_j$ for any $i\in\N$), and codages of integers being the binary words out of $\Fact_{\{0,1\}}(\{10^k\})$. The formulas to define the list $\mF$ of codages of integers are then to be replaced by the following ones, the operator $\underline{\boldsymbol{\cdot}}$ on lists of words being now the one that suppresses all words in $\Fact_{\{0,1\}}(\{10^k\})$:

\begin{itemize} 

\item for the lazy $k$-bonacci numeration system:

\begin{itemize}
\item either $\displaystyle\mF\coloneqq\sum_{\ell\in\N}\mM_\ell$ with $\mM_0\coloneqq\{\varepsilon\}$ and $\displaystyle\mM_{\ell+1}\coloneqq\sum_{i\in\rrbracket \min(k,\ell+1),0\rrbracket}10^{i}\mM_{\ell-i}$ for $\ell\in\N$

\item or $\lim_\ell(\mF_\ell)$, where $\mF_0=\{\varepsilon\}$ and $\displaystyle\mF_{\ell+1}\coloneqq{\underline{\boldsymbol{\cdot}}}\left(0\mF_{\ell}+1\mF_\ell\right)$ for $\ell\in\N$;
\end{itemize}

\item for the $k$-bonacci Gray code (see again \cite{Bernini2013}):

\begin{itemize}

\item either $\displaystyle\mF\coloneqq\sum_{\ell\in\N}\mM_\ell$ with $\mM_0\coloneqq\{\varepsilon\}$ and $\displaystyle\mM_{\ell+1}\coloneqq\sum_{i\in\llbracket 0,\min(k,\ell+1)\llbracket}10^i\overleftarrow{\mM_{\ell-i}}$ for $\ell\in\N$

\item or (as will be proved in Proposition \ref{ZFibo}) $\lim_\ell(\mF_\ell)$, where $\mF_0=\{\varepsilon\}$ and $\displaystyle\mF_{\ell+1}\coloneqq{\underline{\boldsymbol{\cdot}}}\left(0\mF_{\ell}+1\overleftarrow{\mF_{\ell}}\right)$ for $\ell\in\N$.

\end{itemize}

\end{itemize}

(Note that since $\min(k,\ell)\geqslant 0$, the notation $\llbracket \min(k,\ell+1),0\llbracket$ stands for the decreasing lists of integers from $\min(k,\ell+1)$ (included) to $0$ (excluded).)

The links between these two last equivalent definitions will be investigated in the end of Section \ref{ZGeometricSequences}.

\subsection{Organization of the article}\label{sec:OrgArticle}

Section \ref{Basic} set up some notations and definitions, and recall some simple facts on words and lists. Section \ref{GrayProd} adapts the tools of Section \ref{Basic} to the context of Gray codes, allowing us to make algebraic computation on lists. In this section is introduced the {\em Gray product} of two lists (Definition \ref{DefGrayProduct}), which is obtained by a modification of the standard product of lists by incorporating Gray's mirroring process. This gives rise to notions naturally related to it, like the Gray order on words (Definition \ref{DefGrayOrder}) originally introduced in \cite{Vajnovszki} or {\em Gray cartesian product} of lists (Definition \ref{GrayDirectProduct}). The {\em $\mZ$-Gray product} of two lists is then defined as their Gray product in which are removed every word in $\Fact_{\mA}(\mZ)$. 
An important difference between the usual product of lists and the Gray product is that the $\mZ$-counterpart of the latter is not associative in general, contrarily to the usual product. In particular, the infinite list built as a limit of $\mZ$-gray product of a well-chosen initial list is not always power-associative, even if it is the case in the important case of the standard $b$-ary Gray code (Theorem \ref{ExpbAryGray}). Theorem \ref{ClassExp}, one of the two main results of the present paper, provides a quite general class of lists for which power associativity is still satisfied. (This opens the door to the right-concatenation on words as a quite general way to produce Gray codes. This alternative point of view, which can be said {\em mirrorless}, i.e. which does not rely on any mirroring process, will be investigated in the forthcoming paper \cite{Rittaud3}.) Eventually, Theorem \ref{FlidGen}, the other main result of the present paper, extending to many other cases the known examples of standard $b$-ary and $k$-bonacci Gray codes for the flipping digit property.

\section{General tools}\label{Basic}

In this section, we provide precise definitions and statements for the tools we use, paying particular attention to notations to fix them once and for all and in a way that appears the most convenient for our purpose: numbers and words are written in lower latin letters, the latter being indexed starting preferentially from $0$; letters of a word are written in greek lowercase letters (typically: the word $w=(\omega_i)_{i\in\N}$); lists are also indexed starting preferentially from $0$ and written (as well as sets) in calligraphic fonts (an alphabet $\mA$, a list $\mathcal{W}\coloneqq(w_n)_{n\in\N}$ of words\ldots); functions are written in uppercase latin letters (as for the Hamming distance $H$). The notation $\N$, $\N^*$, $\N\cup\{+\infty\}$, etc. stands for increasing lists (not sets). For any $p$, $q\in\R\cup\{+\infty\}$, $\llbracket p,q\llbracket$ stands for the ordered set of integers between $p$ and $q$ (including or not $p$ and/or $q$ depending of the orientation of the corresponding bracket), in increasing order if $p< q$ and in decreasing order if $p> q$.

\subsection{Words and classes of words}\label{Words}

Here, an {\em alphabet} $\mA$ will be an ordered subinterval of $\N$, containing $0$ and $1$. Hence, it can be seen as a list of the form $\llbracket 0,b\llbracket$ for some $b\in\llbracket 2,+\infty\rrbracket$. Elements of $\mA$ are the {\em letters}, or {\em digits} (depending on whether numerical operations are to be done on them or not).

A {\em word} $w$ (on a given alphabet $\mA$) is a sequence in $\mA$. In general it will be written $(\omega_i)_{i\in\llbracket 0,\ell\llbracket}$, where $\ell\in\N\cup\{+\infty\}$ is its {\em length}, also written $|w|$. When $\ell=0$, $w$ is the {\em empty word}, denoted by $\varepsilon$. When $\ell\notin\{0,+\infty\}$ we also write $w$ on the form $\omega_{\ell-1}\cdots\omega_0$ (to keep the usual way of writing numbers, in which the biggest order of magnitude are to the left), $\omega_{\ell-1}$ (resp. $\omega_0$) being the {\em leftmost} (resp. {\em rightmost}) letter of $w$. 

For any $\llbracket p,q\llbracket\ \subset\llbracket 0,\ell\llbracket$, the subsequence $(\omega_i)_{i\in\llbracket p,q\llbracket}$ is a {\em factor} of the word $(\omega_i)_{i\in\llbracket 0,\ell\llbracket}$. For $p=0$ we talk of {\em right factor}, also written $[w]_q$. For $q=\ell$ we talk of a {\em left factor}, also written $[w]^p$.

For two words $w=(\omega_i)_{i\in\llbracket 0,\ell\llbracket}$ and $w'=(\omega_{i+\ell})_{i\in\llbracket 0,\ell'\llbracket}$ with $\ell<+\infty$, the {\em left-concatenation of $w'$ to $w$} or, equivalently, the {\em right-concatenation of $w$ to $w'$}, is the word $w'w\coloneqq(\omega_i)_{i\in\llbracket 0,\ell+\ell'\llbracket}$. In particular, we have $w=[w]^i[w]_i$ for any $i\in\llbracket 0,\ell\llbracket$. Of course, the word $x$ is a factor of the word $w$ iff there exists two words $w'$ and $w''$ such that $w=w'xw''$.

Since the concatenation is associative, we can write $w^n$ for the concatenation of $n\in\N\cup\{+\infty\}$ copies of a finite word $w$ and get $w^{m+n}=w^mw^n$ for any $m$, $n\in\N$, with the convention $w^0=\varepsilon$.

\begin{definition} The equivalence relation $\sim$ on finite words is defined by $0w\sim w$ for any $w\in\mA^*$ and the reflexivity, symmetry and transitivity properties. Put differently, $w\sim w'$ iff there exists $k\in\N$ such that $w=0^kw'$ or $w'=0^{k}w$. 
The set of classes is written $\mA^*_\sim$. The {\em standard representative} of a class is its element of shortest length (or, equivalently, the one with no {\em leading zeroes}). The {\em magnitude} of the finite word $w$, denoted by $\|w\|$, is the length of the standard representative of its class. We therefore have $\|w\|\leqslant |w|$ for any finite word $w$.

\end{definition}

\begin{remark} The previous definition would apply to infinite words $w$ as well (for which the equivalence class is simply reduced to the singleton $\{w\}$), but we will not need it here.
\end{remark}

\begin{definition}
The {\em radix order} is the total order on $\mA^*_\sim$ denoted by $\leqslant $ (the same way as the order in $\mA$) and such that, for $(\omega_i)_{i\in\N}\eqqcolon w\neq w'\coloneqq(\omega'_i)_{i\in\N}$:
\[w< w'\Longleftrightarrow \text{for $j\coloneqq\max\big(i\in\N :\ \omega_i\neq\omega'_i\big)$ we have $\omega_j<\omega'_j$}.\]

Equivalently, for $\ell\coloneqq\max(i\in\N\ :\ \omega_i+\omega'_i>0)$ we have
\[w< w'\Longleftrightarrow \left\{\begin{array}{l}\omega_\ell<\omega'_\ell \\ \text{or}\\ \omega_\ell=\omega'_\ell\text{ and }[w]_{\ell-1}< [w']_{\ell-1}\end{array}\right..\]

Also, equivalently, we have the following characterization of the radix order, closer to the usual way to see it: 
\[w<w'\Longleftrightarrow\left\{\begin{array}{l}\|w\|<\|w'\|\\ \text{or} \\ \|w\|=\|w'\|\text{ and $w$ is smaller than $w'$ for the lexicographical order}\end{array}\right..\]

\end{definition}

\begin{remark} There is a natural bijection between $\{0,1\}^*_\sim$ and finite subsets of $\N$, given by the one-to-one map $(\omega_i)_{i\in\N}\longmapsto\{i\in\N\ :\ \omega_i=1\}$. Confusing two distinct elements $w$ and $w'$ of $\{0,1\}^*_\sim$ with their image by this bijection, the biggest of the two for the radix order is therefore the one containing $\max(w\Delta w')$, where $\Delta$ stands for the symmetric difference between sets.

To extend this remark to $\llbracket 0,b\llbracket^*_\sim$, we can define $w\Delta w'$ as $(1-\delta_{\omega_i,\omega'_i})_{i\in\N}$ (Kronecker's $\delta$) and identifying it with a finite subset of $\N$. Then, we have $w<w'$ iff, for $i\coloneqq\max(w\Delta w')$, we have $\omega_i<\omega'_i$ (where $w=(\omega_i)_{i\in\N}$ and $w'=(\omega'_i)_{i\in\N}$).

\end{remark}

\subsection{Lists}\label{Lists}

We deal first with lists in general, then focus on lists of elements of $\mA^*_\sim$. Most of the time (but not always), our lists do not contain the same element twice, i.e. are in fact ordered sets. 

A {\em list} is a sequence that can be written on the form $(a_n)_{n\in\llbracket 0, p\llbracket}$ where $p\in\N\cup\{+\infty\}$, and the $a_n$ are elements of some predefined set. It is said to be finite or infinite depending on $p$. 
A {\em sublist} of $(a_n)_{n\in\llbracket 0, p\llbracket}$ is a subsequence of $(a_n)_{n\in\llbracket 0, p\llbracket}$ preserving its order. It can be written with the inclusion sign $\subset$. A sublist of the form $(a_n)_{n\in\llbracket q,q'\llbracket}$ with $\llbracket q,q'\llbracket\ \subset\llbracket 0,p\llbracket$ is an {\em interval} of the initial list.

For any nonempty list $\mW$, $\Init(\mW)$ (resp. $\Last(\mW)$) stands for the initial (resp. last) element of the list $\mW$. For any $w\in\mW$, we write $\suc_\mW(w)$ and $\pre_\mW(w)$ for the successor and predecessor of $w$ in $\mW$ (if any).

For any finite lists $\mathcal{W}\coloneqq(w_n)_{n\in\llbracket 0,p\llbracket}$ and $\mathcal{W}'\coloneqq(w_{n+p})_{n\in\llbracket 0,p'\llbracket}$ and for any integer $d$ we define
\[\begin{array}{rcl}
\overleftarrow{\mathcal{W}}&\coloneqq&(w_{p-n})_{n\in\rrbracket 0,p\rrbracket}=(w_n)_{n\in\rrbracket p,0\rrbracket},\\
\mathcal{W}^{\leftrightarrow d}&\coloneqq&\begin{cases}\mW&\text{if $d$ is even;}\\ \overleftarrow{\mW}&\text{if $d$ is odd.}\end{cases}\\
\mathcal{W}+\mathcal{W'}&\coloneqq&(w_n)_{n\in\llbracket0,p+p'\llbracket}.
\end{array}\]

Since the sum of lists is an associative binary operation with neutral element $\varnothing$, we can define univoqually the (noncommutative) sum 
$\dy\sum_{i\in\llbracket 0,m\llbracket}\mathcal{W}_{i}$ of any sequence $(\mathcal{W}_{i})_{i\in\llbracket 0,m\llbracket}$ of finite lists (with $m\in\N\cup\{+\infty\}$) in a natural way, the empty sum being equal to the empty list $\varnothing$. 
We then define the {\em usual cartesian product} of $\mW$ and $\mW'$ as:
\[\mW\times\mW'\coloneqq\sum_{n\in\llbracket 0,p\llbracket}\sum_{n'\in\llbracket 0,p'\llbracket}(w_n,w'_{n'}).\]


Now, consider lists in $\mA^*_\sim$. The following property will be satisfied by most of these considered in this paper.

\begin{definition}\label{MagnitudeIncreasing} A list $\mathcal{W}=(w_n)_{n\in\llbracket 0,p\llbracket}$ in $\mA^*_\sim$ for which the sequence $(\|w_n\|)_{n\in\llbracket 0,p\llbracket}$ is increasing is said to be {\em magnitude increasing}.
\end{definition}

The list $\mathcal{W}$ being given, we generically write $\mathcal{M}_\ell$ for its sublist made of its elements of magnitude exactly $\ell$. Hence, $\mW$ is magnitude increasing iff $\mW=\sum_{\ell\in\N}\mM_\ell$.


The notation $[{\mathcal{W}}]_i$ will have two possible meanings (the one to be considered being explicited in each context): its {\em cardinality meaning}, for which $[{\mathcal{W}}]_i\coloneqq(w_n)_{n\in\llbracket 0,i\llbracket}$ with $i\in\llbracket 0,p\llbracket$ (a notation that can be used also for a list made of other elements than words), and its {\em length meaning}, for which $[{\mathcal{W}}]_i\coloneqq([w_n]_i)_{n\in\llbracket 0,p\llbracket}$ with $i\in\N$.

For lists, two different meanings of the notion of {\em increase} are to be distinguished. A given list in $\mA^*_\sim$ is {\em increasing} whenever the sequence it defines is increasing under some specified order relation; a sequence $(\mathcal{W}_\ell)_{\ell\in\N}$ of finite lists in $\mA^*_\sim$ is {\em increasing}  whenever, for any $\ell\in\N$,  there exists a list $\mX_{\ell+1}$ such that $\mathcal{W}_{\ell+1}=\mathcal{W}_{\ell}+\mX_{\ell+1}$. The {\em limit} of such an increasing sequence is then defined by $\lim_\ell(\mathcal{W}_\ell)=\mW_0+\sum_{\ell\in\N^*}\mX_\ell$.

\subsection{Ordinary product of lists of words}\label{ProductLists}

Let $\mathcal{W}\coloneqq(w_n)_{n\in\llbracket 0,p\llbracket}$ be a list in $\mA^*_\sim$, and put $\|\mathcal{W}\|\coloneqq\max_{n\in\llbracket 0,p\llbracket}(\|w_n\|)$, with $\|\varnothing\|\coloneqq 0$. If $\|\mW\|$ is finite, for any $w\in\mA^*_\sim$ we define $w\mathcal{W}\coloneqq (w[w_n]_{\|\mathcal{W}\|})_{n\in\llbracket 0,p\llbracket}$, the bracket notation having its length meaning.


\begin{definition}\label{ProductOfLists} Let $\mathcal{W}\coloneqq(w_n)_{n\in\llbracket 0,p\llbracket}$ and $\mathcal{W}'\coloneqq(w'_{n'})_{n'\in\llbracket 0,p'\llbracket}$ be two lists in $\mA^*_\sim$ with $\mathcal{W}'$ finite. Their (noncommutative) {\em ordinary product} $\mathcal{W}\mathcal{W'}$ (or $\mW\boldsymbol{\cdot}\mW'$) is the list:
\[\mW\boldsymbol{\cdot}\mW'=\mathcal{W}\mathcal{W'}\coloneqq\sum_{n\in\llbracket 0,p\llbracket}w_n\mathcal{W}'=\sum_{(n,n')\in\llbracket 0,p\llbracket\times \llbracket 0,p'\llbracket}w_n[w'_{n'}]_{\|\mathcal{W}'\|}.\]
The list $\{\varepsilon\}$ is the neutral element for this product, and $\varnothing$ its absorbing element.
\end{definition}

\begin{remark} In practice, the normalization of the lengths of representatives choosed for the elements of $\mW'$ in Definition \ref{ProductOfLists} makes it handy, for reading as well as computer programming, to consider $\mW'$ as a list of words of equal length, appending leading zeroes to those of magnitude less than $\|\mW'\|$. For ease of reading we will thus be led to write, for example, $0\mW+1\mW$ instead of $\mW+1\mW$ for the result of the product $\{0,1\}\boldsymbol{\cdot}\mW$.
\end{remark}

\begin{proposition}\label{ProductListAssociative} The ordinary product of lists is associative.
\end{proposition}

\begin{proof} For $\mathcal{W}''=(w''_{n''})_{n''\in\llbracket 0,p''\llbracket}$, we have
\begin{eqnarray*}
(\mathcal{W}\mathcal{W}')\mathcal{W}''&=&\sum_{(n,n')\in\llbracket 0,p\llbracket\times\llbracket 0,p'\llbracket}w_n[w'_{n'}]_{\|\mathcal{W}'\|}\mathcal{W}'' \\
&=& \sum_{(n,n',n'')\in\llbracket 0,p\llbracket\times\llbracket 0,p'\llbracket\times \llbracket 0,p''\llbracket} w_n[w'_{n'}]_{\|\mathcal{W}'\|} [w''_{n''}]_{\|\mathcal{W}''\|}
\end{eqnarray*}
\noindent and
\begin{eqnarray*}
\mathcal{W}(\mathcal{W}'\mathcal{W}'')&=&\sum_{n\in\llbracket 0,p\llbracket}w_n(\mathcal{W}'\mathcal{W}'')\\
& =& \sum_{n\in\llbracket 0,p\llbracket} w_n \left[\sum_{(n',n'')\in\llbracket 0,p'\llbracket\times \llbracket 0,p''\llbracket} w'_{n'}[w''_{n''}]_{\|\mathcal{W}''\|}\right]_{\|\mathcal{W}'\mathcal{W}''\|}\\
& =& \sum_{n\in\llbracket 0,p\llbracket} w_n \sum_{(n',n'')\in\llbracket 0,p'\llbracket\times \llbracket 0,p''\llbracket} \left[w'_{n'}[w''_{n''}]_{\|\mathcal{W}''\|}\right]_{\|\mathcal{W}'\mathcal{W}''\|},
\end{eqnarray*}
\noindent from which the result follows since $\|\mathcal{W}'\mathcal{W}''\|=\|\mathcal{W}'\|+\|\mathcal{W}''\|$.
\end{proof}

The set of lists of words on a given alphabet is therefore equipped with two binary operators, sum and ordinary product of lists, which are both associative (and noncommutative). The ordinary product of lists being associative, it is also {\em power associative}, i.e. for any finite list of words $\mathcal{W}$ we have $\mathcal{W}^{\boldsymbol{\cdot}(r+s)}=\mathcal{W}^{\boldsymbol{\cdot} r}\mathcal{W}^{\boldsymbol{\cdot} s}$ for any natural integers $r$ and $s$, where $\mathcal{W}^{\boldsymbol{\cdot}(m+1)}\coloneqq\mathcal{W}\mathcal{W}^{\boldsymbol{\cdot}m}$ for any $m\in\N$ and $\mathcal{W}^{\boldsymbol{\cdot}0}\coloneqq\{\varepsilon\}$.

As an example, the codage of the standard $b$-ary numeration system can be seen as the increasing (thus convergent) sequence of lists $(\mathcal{B}_\ell)_{\ell\in\N}$ in $\mA^*_\sim$ (with $\mA=\llbracket 0,b\llbracket$) by:
\[\mathcal{B}_0\coloneqq\{\varepsilon\}\qquad\mathcal{B}_{\ell+1}\coloneqq\mA\boldsymbol{\cdot}\mathcal{B}_{\ell}\text{ for $\ell\in\N$},\]
\noindent so $\mB_\ell=\mA^{\boldsymbol{\cdot}\ell}$ for any $\ell\in\N$. Hence, for any integers $\ell$ and $\ell'$ we have $\mB_{\ell+\ell'}=\mB_\ell\mB_{\ell'}$. An algorithm of exponentiation by squaring can thus be applied to reach the full list of $b$-ary codages as the limit of the subsequence $(\mathcal{B}'_m)_{m\in\N^*}$ defined by $\mathcal{B}'_1\coloneqq\mA$ and $\mathcal{B}'_{m+1}\coloneqq\mathcal{B}'_m\mathcal{B}'_m$ for any $m\in\N^*$.

\subsection{The $\mZ$-ordinary product}\label{sec:ZProduct}

\begin{definition}\label{ZProduct} Let $\mA$ be an alphabet and let $\mZ\subset\mA^*$. We define $\underline{\boldsymbol{\cdot}}\mW$ as the sublist of $\mW$ in which is removed every element $w$ of $\mW$ such that $[w]_{\|\mW\|}\in\Fact_\mA(\mZ)$. For $\mW'$ be another list, the {\em $\mZ$-ordinary product} of $\mW$ and $\mW'$ is the list $\underline{\boldsymbol{\cdot}}(\mW\mW')$, also written $\mW\underline{\boldsymbol{\cdot}}\mW'$.\end{definition}

\begin{remark} We can thus write $\underline{\boldsymbol{\cdot}}\mW=\{\varepsilon\}\underline{\boldsymbol{\cdot}}\mW$.
\end{remark}

\begin{remark} It is only when $\mZ$ does not contain any word with $0$ as a leftmost letter that $\underline{\boldsymbol{\cdot}}\mW$ is the sublist of $\mW$ made of its elements not in $\Fact_\mA(\mZ)$.
\end{remark}

\begin{proposition}\label{ZProductAssociatif} The $\mZ$-ordinary product is associative.
\end{proposition}

\begin{proof} Observe that, for any lists $\mW$, $\mW'$ and $\mW''$ we have $\mW\underline{\boldsymbol{\cdot}}(\mW'\underline{\boldsymbol{\cdot}}\mW'')=\mW\underline{\boldsymbol{\cdot}}(\mW'\boldsymbol{\cdot}\mW'')$ as well as $(\mW\underline{\boldsymbol{\cdot}}\mW')\underline{\boldsymbol{\cdot}}\mW''=(\mW{\boldsymbol{\cdot}}\mW')\underline{\boldsymbol{\cdot}}\mW''$. We can thus write, by the associativity of the ordinary product of lists:
\begin{eqnarray*}
\mW\underline{\boldsymbol{\cdot}}(\mW'\underline{\boldsymbol{\cdot}}\mW'')&=&\mW\underline{\boldsymbol{\cdot}}(\mW'{\boldsymbol{\cdot}}\mW'')\\
&=&\underline{\boldsymbol{\cdot}}\Big(\mW\boldsymbol{\cdot}(\mW'\boldsymbol{\cdot}\mW'')\Big)\\
&=&\underline{\boldsymbol{\cdot}}\Big((\mW\boldsymbol{\cdot}\mW')\boldsymbol{\cdot}\mW''\Big)\\
&=&(\mW\boldsymbol{\cdot}\mW')\underline{\boldsymbol{\cdot}}\mW''\\
&=&(\mW\underline{\boldsymbol{\cdot}}\mW')\underline{\boldsymbol{\cdot}}\mW''.
\end{eqnarray*}
\end{proof}

In a similar way to the ordinary product of lists, the associativity of the $\mZ$-ordinary product makes it power-associative, so for any list $\mW$ and any $r$, $s\in\N$ we have that $\mW^{\underline{\boldsymbol{\cdot}}(r+s)}=\mW^{\underline{\boldsymbol{\cdot}} r}\mW^{\underline{\boldsymbol{\cdot}}s}$, where $\mW^{\underline{\boldsymbol{\cdot}}(r+1)}\coloneqq\mW\underline{\boldsymbol{\cdot}}\mW^{\underline{\boldsymbol{\cdot}}r}$ and $\mW^{\underline{\boldsymbol{\cdot}}0}\coloneqq\{\varepsilon\}$.

The lazy $k$-bonacci numeration system is a natural example of utilization of the $\mZ$-ordinary product. Put $\mZ\coloneqq\{10^k\}$ and define the sequence $(\mF_\ell)_{\ell\in\N}$ as $\mF_\ell\coloneqq\{0,1\}^{\underline{\boldsymbol{\cdot}}\ell}$. It is easy to check that the list $\mF\coloneqq(f_n)_{n\in\N}$ defined as the limit of $(\mF_\ell)_{\ell\in\N}$ is such that, for any $n\in\N$, the word $f_n$ is the lazy $k$-bonacci expansion of $n$. Also, for any integers $\ell$ and $\ell'$ we have $\mF_\ell\underline{\boldsymbol{\cdot}}\mF_{\ell'}=\mF_{\ell+\ell'}$.

\begin{remark} Take $k=2$, so $\#(\mF_\ell)=z_\ell$ for all $\ell\in\N$, where $(z_\ell)_{\ell}$ is the Zeckendorf sequence ($z_0=1$, $z_1=2$ and $z_{\ell}=z_{\ell-1}+z_{\ell-2}$ for all $\l\geqslant 2$, see Section \ref{IntroOperations}). For $\ell$ and $\ell'$ at least equal to $2$, the elements of $\mF_\ell {\boldsymbol{\cdot}} \mF_{\ell'}\cap\Fact_{\{0,1\}}(\{11\})$ are those of the form $f0110f'$, where $f01\in\mF_\ell$ and $10f'\in\mF'_{\ell'}$. Hence, we have $\#(\mF_\ell\underline{\boldsymbol{\cdot}}\mF_{\ell'})=\#(\mF_\ell{\boldsymbol{\cdot}}\mF_{\ell'})-\#(\mF_{\ell-2})\#(\mF_{\ell'-2})$, from which we obtain the well-known equality $z_{\ell+\ell'}=z_{\ell}z_{\ell'}-z_{\ell-2}z_{\ell'-2}$ (more oftenly written for the standard Fibonacci sequence, that differs from the Zeckendorf sequence by a double shift of indexation), that remains true for lower values of $\ell$ and $\ell'$ by putting $z_{\text{-}1}\coloneqq 1$ and $z_{\text{-}2}\coloneqq 0$. Of course, this remark can be extended to get similar formulas from $k$-bonacci numeration systems.
\end{remark}

More generally, we can use the $\mZ$-ordinary product to get a large part of sequences of words defined by some combinatorial constraint, in particular numeration systems defined by linear recurring sequences.

\section{Gray and $\mZ$-Gray products}\label{GrayProd}

\subsection{The Gray product and $\Flid(\mI)$}\label{sec:GrayProd}

\begin{definition}\label{Hamming} Let $w\coloneqq(\omega_i)_{i\in\N}$ and $w'\coloneqq(\omega'_i)_{i\in\N}$ be two elements of $\mA^*_\sim$. Their {\em Hamming distance} is the value $H(w,w')\coloneqq\#\{i\in\N\ :\ \omega_i\neq\omega'_i\}$.
\end{definition}

As recalled in Section \ref{IntroGrayCodesConnus}, Gray's initial goal was to list all finite binary words such that two consecutive ones in the list differ by a single digit. We consider the following generalization of this notion, Gray's initial one corresponding to $\mI=\{1\}$:

\begin{definition}\label{FlippingDigitProperty}
Let $\mathcal{W}$ be a list in $\mA^*_\sim$. We define, with the cardinality meaning, $H(\mW)\coloneqq\{H(w,\suc_\mW(w))\ :\ w\in[\mW]_{\#(W)-1}\}$. For any $\mI\subset\N$, we write $\Flid(\mI)$ for the set of lists $\mW$ for which $H(\mW)\subset\mI$ (i.e.  {\em satisfies the $\mI$-flipping digit property}).

If $\mW$ is finite with at least two elements and satisfies that $H(\Init(\mW),\Last(\mW))\in H(\mW)$, then $\mathcal{W}$ is {\em cyclical}.

\end{definition}

Extending Gray's initial construction quite naturally leads to the following definition:

\begin{definition}\label{DefGrayProduct} The {\em Gray product} of two finite lists $\mathcal{W}\coloneqq(w_n)_{n\in\llbracket 0,p\llbracket}$ and $\mathcal{W}'$ in $\mA^*_\sim$ is the list:
\[\mathcal{W}\odot\mathcal{W}'\coloneqq\sum_{n\in\llbracket 0,p\llbracket}w_n\mathcal{W}'^{\leftrightarrow n}.\]
\end{definition}

The previous definition also makes sense for $\mW$ infinite, but we will not need it here.


\begin{proposition}\label{GrayProdFlid} For any finite lists $\mW$ and $\mW'$ we have $H(\mW\odot\mW')= H(\mW)\cup H(\mW')$. If $\mW$ is cyclical and $\#(\mW')\in 2\N$, then $\mW\odot\mW'$ is cyclical.



\end{proposition}

The proof is trivial.

Applying Proposition \ref{GrayProdFlid} to $\mW=\llbracket 0,b\llbracket$ and $\mW'=\mG_\ell$ for some $\ell\in\N$ gives by a simple induction the following well-known

\begin{corollary}\label{bAryGrayFlid1} For any $b\geqslant 2$ we have $H(\mG)=\{1\}$. Moreover, $\mG$ is cyclical iff $b$ is even.
\end{corollary}


Of course, applications of Proposition \ref{GrayProdFlid} go beyond Corollary \ref{bAryGrayFlid1}. For example, consider two lists $\mathcal{W}$ and $\mathcal{W}'$ in $\Flid(\mI)$, put $\mX_0\coloneqq\{\varepsilon\}$ and, for any $\ell\in\N^*$, define $\mathcal{X}_{\ell}$ as either $\mW\odot\mathcal{X}_{\ell-1}$ or $\mW'\odot\mathcal{X}_{\ell-1}$, the choice being made arbitrarily for each $\ell$. The sequence $(\mX_\ell)_{\ell\in\N}$ is thus made of terms in $\Flid(\mI)$. If we want, moreover, $(\mathcal{X}_\ell)_{\ell\in\N}$ to converge to an infinite list, it is sufficient that $\mW$ and $\mW'$ have at least two elements each, the initial one being $\varepsilon$.

To go further, we first need more properties of the Gray product. The ones given in the next proposition are routinely proved.

\begin{proposition}\label{DistribGrayProd} For any finite lists of words $\mathcal{W}$, $\mathcal{W}'$ and $\mathcal{W}''$ we have
\[(\mathcal{W}+\mathcal{W}')\odot \mathcal{W}''=\mathcal{W}\odot\mathcal{W}'' + \mathcal{W}'\odot\mathcal{W}''^{\leftrightarrow\#(\mathcal{W})},\]
\[\overleftarrow{\mathcal{W}\odot\mathcal{W}'}=\overleftarrow{\mathcal{W}}\odot\mathcal{W}'^{\leftrightarrow\#(\mathcal{W})}.\]
\end{proposition}

\begin{remark} For the purpose of the second proof of Proposition \ref{GrayProdAssociatif}, it is useful to write the following ``generalization'' of the latter equality to any $k\in\N$ (which, of course, reduces to the case $k=0$):
\[(\mathcal{W}\odot\mathcal{W}')^{\leftrightarrow k}=\mathcal{W}^{\leftrightarrow k}\odot\mathcal{W}'^{\leftrightarrow k\#(\mathcal{W})}.\]
\end{remark}

The following result has important consequences.

\begin{proposition}\label{GrayProdAssociatif} The Gray product is associative.
\end{proposition}

\begin{proof} We prove the equality $(\mW\odot\mW')\odot\mW''=\mW\odot(\mW'\odot\mW'')$ by induction on $p\coloneqq\#(\mW)$. If $p=0$ then $\mW=\varnothing$ which is the absorbing element, so the desired equality is satisfied whatever $\mW'$ and $\mW''$ are. Now, let assume that the equality hold for any $\mW$, $\mW'$ and $\mW''$ with $\#(\mW)=p$, and consider a triple $(\mW,\mW',\mW'')$ with $\mW=p+1$. Writing $\mathcal{W}\eqqcolon\mathcal{X}+\{w\}$ (assuming $\|\mathcal{X}\|=\|w\|$, without real loss of generality by the remark following \ref{ProductOfLists}), we have $\#(\mX)=p$ so, by the induction hypothesis, we can write $(\mathcal{X}\odot\mathcal{W}')\odot\mathcal{W}''=\mathcal{X}\odot(\mathcal{W}'\odot\mathcal{W}'')$. Put $p'\coloneqq\#(\mathcal{W}')$. With the help of the remark after Proposition \ref{DistribGrayProd} we have
\begin{eqnarray*}
(\mathcal{W}\odot\mathcal{W}')\odot\mathcal{W}''&=&\big((\mathcal{X}+\{w\})\odot\mathcal{W}'\big)\odot\mathcal{W}''\\
&=&\big(\mathcal{X}\odot\mathcal{W}'+w\mathcal{W}'^{\leftrightarrow p}\big)\odot\mathcal{W}''\\
&=&(\mathcal{X}\odot\mathcal{W}')\odot\mathcal{W}''+w\mathcal{W}'^{\leftrightarrow p}\odot\mathcal{W}''^{\leftrightarrow pp'}\\
&=&\mathcal{X}\odot(\mathcal{W}'\odot\mathcal{W}'')+w(\mathcal{W}'\odot\mathcal{W}'')^{\leftrightarrow p}\\
&=&(\mathcal{X}+\{w\})\odot(\mathcal{W}'\odot\mathcal{W}'')\\
&=&\mathcal{W}\odot(\mathcal{W}'\odot\mathcal{W}''),
\end{eqnarray*}
\noindent so we are done.
\end{proof}

We can therefore define a notion of $\odot$-exponentiation:

\begin{corollary}\label{ExponentiationW} Let $\mathcal{W}$ be a finite list in $\mA^*_\sim$. For any $\ell\in\N$, define $\mathcal{W}^{\odot\ell}$ as $\mathcal{W}^{\odot 0}\coloneqq\{\varepsilon\}$ and, for $\ell\in\N^*$, $\mathcal{W}^{\odot\ell}\coloneqq\mathcal{W}\odot\mathcal{W}^{\odot(\ell-1)}$. For any $\ell$, $\ell'\in\N$ we have $\mathcal{W}^{\odot(\ell+\ell')}=\mathcal{W}^{\odot\ell}\odot\mathcal{W}^{\odot\ell'}$.
\end{corollary}

In the important case of $\mW$ containing $\varepsilon$ and not containing the same element twice, a recodage allows to rewrite it as an alphabet, so in this case the previous corollary is no more than the following crucial one, an application of which will be given in Section \ref{GrayOrder} (Theorem \ref{ExpbAryGray} and consequences):

\begin{corollary}\label{Exponentiation} Let the integer $b\geqslant 2$ be given, defining the standard $b$-ary Gray code $(\mG_\ell)_{\ell\in\N}$. For any $\ell$, $\ell'\in\N$ we have ${\mG}^{\odot(\ell+\ell')}={\mG}^{\odot\ell}\odot{\mG}^{\odot\ell'}$.

\end{corollary}

As an application of Corollary \ref{Exponentiation}, we have the following alternative definition of the standard $b$-ary Gray code by right-appending letters, generalizations of which will be investigated in the forthcoming paper \cite{Rittaud3}:

\begin{theorem}\label{bAryRight} For $b\geqslant 2$ be a given integer, the sequence $(\mG_\ell)_{\ell\in\N}$ defining the standard $b$-ary Gray code can be defined in the following way. Start with $\mG_0=\{\varepsilon\}$ and assume $\mG_{\ell}$ is built for some $\ell\in\N$. Write each element of $\mG_{\ell}$ $b$ times in a row, then right-append one letter to each element of this new list, going from $0$ to $b-1$, then from $b-1$ to $0$, then again from $0$ to $b-1$, etc. The list thus obtained is equal to $\mG_{\ell+1}$.
\end{theorem}

\begin{proof} Simple application of Corollary \ref{Exponentiation} with $\ell'=1$, i.e. $\mG_{\ell+1}=\mG_{\ell}\odot\mA$.
\end{proof}

Also, we can set up for the standard $b$-ary Gray code an algorithm of exponentiation by squares to get $\lim_\ell(\mG_\ell)$ as the limit of the subsequence $(\mG_{2^m})_{m\in\N}$ by computing $\mG_{2^m}=\mG_{2^{m-1}}\odot\mG_{2^{m-1}}$ for $m\in\N^*$. The sequences $(\mathcal{X}_\ell)_{\ell\in\N}$ introduced after Proposition \ref{GrayProdFlid} is another example of application of Corollary \ref{Exponentiation}.

To obtain by this kind of process the $k$-bonacci Gray codes presented in Section \ref{IntroOperations} as well as other generalizations, the coming section will adapt the Gray product in a similar way the ordinary product of lists was adapted to obtain the lazy $k$-bonacci numeration system.

\subsection{The Gray order}\label{GrayOrder}

Until now, we used the notation $\mG$ and $\mG_\ell$ for the standard $b$-ary Gray code, since $b$ was fixed one and for all. In the beginning of the present section, we sometimes use $\mG^{(b)}$ and $\mG^{(b)}_\ell$ as well, when it is necessary to distinguish between different values of $b$ possibly involved.

\begin{proposition}\label{InclusionDesGb} The sequence $(\mG^{(b)})_{b\in\llbracket 2,+\infty\llbracket}$ is increasing for the inclusion of lists.
\end{proposition}

\begin{proof} By induction, we have $\mG^{(b)}_\ell\subset\mG^{(b+1)}_\ell$ for any $b$ and any $\ell$.
\end{proof}

The sequence of lists $(\mG^{(b)})_{b\in\llbracket 2,+\infty\llbracket}$ is not increasing in the sense given in the end of Section \ref{Lists}, but the fact that it is increasing for the order defined by inclusion of lists still allows to consider its set limit, which is simply $\N^*_\sim$ (where, of course, $\N^*$ stands for the set of finite words on $\N$ and not $\N\backslash\{0\}$). Moreover, since the inclusion in Proposition \ref{InclusionDesGb} is order-preserving, the following definition, close to the one introduced in \cite{Vajnovszki}, makes sense:

\begin{definition}\label{DefGrayOrder}  The {\em Gray order}, denoted by $\preccurlyeq$ (or $\prec$ when the inequality is strict), is the total order on $\N^*_\sim$ for which, for any $w$, $w'\in\N^*_\sim$, we have $w\preccurlyeq w'$ iff, for some $b\in\llbracket 2,+\infty\llbracket$ big enough so that $\mG^{(b)}$ contains both $w$ and $w'$, $w$ is before $w'$ in $\mG^{(b)}$.
\end{definition}

In particular, for a list $\mW$ monotonic for the Gray order we can consider {\em Gray intervals} on it: if $x$ and $x'$ are two elements of $\N^*_\sim$ with $x\preccurlyeq x'$, we set $\llbracket x, x'\rrbracket_\mW\coloneqq\{w\in\mW\ :\ x\preccurlyeq w\preccurlyeq x'\}$, and $\llbracket x',x\rrbracket_\mW\coloneqq\overleftarrow{\llbracket x, x'\rrbracket_\mW}$. The following notion shares links with what is called {\em parity problems} in \cite{Squire}:

\begin{definition}\label{IndexParityProperty} Define $\Sigma$\ :\ $\N^*_\sim\longrightarrow\Z/2\Z$ by $\Sigma((\omega_i)_{i\in\N})\coloneqq\left(\sum_{i\in\N}\omega_i\right)\bmod 2$.  A list $\mW=(w_n)_{n\in\llbracket 0,p\llbracket}$ in $\N^*_\sim$ satisfies the {\em index parity property} if, for all $n\in\llbracket 0,p\llbracket$, we have $\Sigma(w_n)=n\bmod 2$.
\end{definition}

\begin{definition}\label{GrayDirectProduct} The {\em Gray cartesian product} of the lists $\mW\coloneqq(w_n)_{n\in\llbracket 0,p\llbracket}$ and $\mW'\coloneqq(w'_n)_{n\in\llbracket 0,p'\llbracket}$ is the following list of pairs:
\[\mW\otimes\mW'\coloneqq\sum_{n\in\llbracket 0,p\llbracket}\sum_{n'\in\llbracket 0,p'\llbracket^{\leftrightarrow n}}(w_n,w'_n).\]

In particular, we have

\[\mW\odot\mW'=\sum_{(i,i')\in\llbracket 0,p\llbracket\otimes\llbracket 0,p'\llbracket}w_i[w'_{i'}]_{\|\mW'\|}.\]

\end{definition}

\begin{proposition}\label{IndexParityPreservationOdot} If $\mW$ and $\mW'$ are two finite lists in $\N^*_\sim$ that both satisfy the index parity property, then so does $\mW\odot\mW'$.
\end{proposition}

\begin{proof} Write $\mW=(w_i)_{i\in\llbracket 0,p\llbracket}$ and $\mW'=(w'_i)_{i\in\llbracket 0,p'\llbracket}$, so that, as in the previous definition:
\[\mW\odot\mW'=\sum_{(i,i')\in\llbracket 0,p\llbracket\otimes\llbracket 0,p'\llbracket}w_i[w'_{i'}]_{\|\mW'\|}.\]

In this list, the index of the element $w_i[w'_{i'}]_{\|\mW'\|}$ given by the pair $(i,i')$ is equal to $p'i+T^i(i')$, where $T$ is the function defined by $T(x)=(p'-1)-x$. By hypothesis on $\mW$ and $\mW'$ we also have, the following equalities being understood modulo $2$:
\[\Sigma\left(w_i [w'_{i'}]_{\|\mW'\|}\right)=\Sigma(w_i)+\Sigma(w'_{i'})=i+i'\]
\noindent and
\[p'i+T^i(i')=i'\mathbf{1}_{2\N}(i)+(1+i')\mathbf{1}_{1+2\N}(i)=i+i',\]
\noindent so $\mW\odot\mW'$ satisfies the index parity property.\end{proof}



\begin{corollary}\label{GellFlid1} For any $b\geqslant 2$, $\mG^{(b)}$ satisfy the index parity property.
\end{corollary}

\begin{proof} Simple induction on $\ell\in\N$ using Proposition \ref{IndexParityPreservationOdot} with $\mW=\llbracket 0,b\llbracket$ and $\mW'=\mG^{(b)}_\ell$.
\end{proof}

Corollary \ref{GellFlid1} shows that $(\Sigma(\mG^{(b)}))=(n\bmod 2)_{n\in\N}$ for any $b\geqslant 2$. We leave as a question to investigate the behaviour of the the sum-of-digits sequence defined by $\mG^{(b)}$.

The following result is a natural consequence of what precedes and will be of great importance.

\begin{theorem}\label{ExpbAryGray} For any integer $b\geqslant 2$, any $\ell\in\N$ and any $\ell'\in\llbracket 0,\ell\rrbracket$ we have
\[\mG_\ell=\sum_{g\in\mG_{\ell'}}g\mG_{\ell-\ell'}^{\leftrightarrow\Sigma(g)}.\]
\end{theorem}

\begin{proof} By Corollary \ref{Exponentiation} we have $\mG_\ell=\mG_{\ell'}\odot\mG_{\ell-\ell'}$, hence the result by Corollary \ref{GellFlid1} applied to $\mG_{\ell'}$.
\end{proof}

As an application of Theorem \ref{ExpbAryGray}, consider the following definition, first introduced in \cite{VW}.

\begin{definition}\label{PrefPart} The list $\mathcal{W}$ in $\N^*_\sim$ is {\em prefix-partitioned} if, for any $\tilde{w}\in\N^*_\sim$ and any $i\in\N$, the sublist $\{w\in\mW\ :\ [w]^i=\tilde{w}\}$ is an interval of $\mW$.

\end{definition}

An first easy fact is that

\begin{theorem}\label{GrayOrderPrefPartGell} For any $b\geqslant 2$, $\mG$ is prefix-partitioned.
\end{theorem}

\begin{proof} Let $\tilde{w}\in\llbracket 0,b\llbracket^*_\sim$ and $i\in\N$. Put $\ell'\coloneqq\|\tilde{w}\|$ and $\ell\coloneqq\ell'+i$. What we have to prove is that $\{w\in\mG\ :\ [w]^i=\tilde{w}\}$ is an interval of $\mG_\ell$. By Theorem \ref{ExpbAryGray} and the definition of the Gray order we have
\[\mG_\ell=\left(\sum_{g\in\llbracket\varepsilon,\tilde{w}\llbracket_{\mG_{\ell'}}}g\mG_{\ell-\ell'}^{\leftrightarrow\Sigma(g)}\right) + \tilde{w}\mG_{\ell-\ell'}^{\leftrightarrow\Sigma(\tilde{w})} + \left(\sum_{g\in\rrbracket\tilde{w},\max(\mG_{\ell'})\rrbracket_{\mG_{\ell'}}}g\mG_{\ell-\ell'}^{\leftrightarrow\Sigma(g)}\right),\]
\noindent so we are done.
\end{proof}

\begin{corollary}\label{GrayOrderPrefPart} Any list in $\N^*_\sim$ which is monotonic for the Gray order is prefix-partitioned.
\end{corollary}

\begin{proof} Without loss of generality, consider a list $\mW$ strictly increasing. Let $\tilde{w}\in\llbracket 0,b\llbracket^*$ for some $b\geqslant 2$, let $i\in\N$, and consider the sublist $\mW'\coloneqq\{w\in\mW\ :\ [w]^i=\tilde{w}\}$. If $\mW'$ is finite, then it is a sublist of $\mG^{(b')}$ for some big enough $b'\geqslant b$. Since a sublist of a prefix-partitioned list is also prefix-partitioned, the result is given by Theorem \ref{GrayOrderPrefPartGell}. Now, if $\mW'$ is infinite, let $y\in\mW$. Since $\llbracket\varepsilon,y\rrbracket_\mW$ is finite, the previous argument applies to $\mW'\cap\llbracket 0,y\rrbracket$, which is therefore an interval for any $y\in\mW$. Hence, $\mW'$ is itself an interval.
\end{proof}

To provide some practical characterizations of the Gray order, let us define first the following notation:

\begin{definition} Let $\leqslant$ be a (total or partial) order on some set (and $<$ the corresponding strict order). For $x$, $y$ in the set and $m\in\Z/2\Z$ (or $m\in\Z$), the notation $x\stackrel{m}{\scriptstyle{<>}}y$ means $x<y$ if $m=0$ (or $m\in2\Z$) and $x>y$ if $m=1$ (or $m\in 1+2\Z$).
\end{definition}

Here are now some ways to know the Gray order between elements of $\N^*_\sim$:

\begin{proposition}\label{CharacterizationsGrayOrder} Let $w=(\omega_i)_{i\in\N}$ and $w'=(\omega'_i)_{i\in\N}$ be two distincts elements of $\N^*_\sim$.

\begin{itemize}

\item Let $k$ be the biggest index for which the letters $\omega_k$ and $\omega'_k$ are not both $0$(i.e. $k\coloneqq\max(\|w\|,\|w'\|)-1$). We have
\[w\prec w'\Longleftrightarrow \Big(\omega_k<\omega'_k\Big)\text{ or }\left(\omega_k=\omega'_k\eqqcolon\omega\text{ and } [w]_{k}\stackrel{\omega}{\scriptstyle{\prec\succ}}[w']_{k}\right).\]

\item Let $j$ be the biggest integer such that $\omega_j\neq\omega'_j$. We have
\[w\prec w'\Longleftrightarrow [w]_{j+1} \stackrel{\Sigma([w]^{j+1})}{\scriptstyle{\prec\succ}}[w']_{j+1}\Longleftrightarrow \omega_{j}\stackrel{\Sigma([w]^{j+1})}{\scriptstyle{<>}}\omega'_{j}.\]

\item For any $i\in\N$ we have
\[w\prec w'\Longleftrightarrow\Big([w]^{i}\prec [w']^{i}\Big)\text{ or }\left([w]^{i}=[w']^{i}\text{ and }[w]_{i}\stackrel{\Sigma([w]^{i})}{\scriptstyle{\prec\succ}}[w']_{i}\right).\]

\end{itemize}
\end{proposition}

\begin{proof} We choose once and for all an integer $b$ such that both $w$ and $w'$ belong to $\mG_\ell=\mG^{(b)}_\ell$, where $\ell\coloneqq k+1$.

By definition of the Gray order, $w\prec w'$ iff $w$ comes before $w'$ in $\mG_\ell=\sum_{\alpha\in\llbracket 0,b\llbracket}\alpha\mG_k^{\leftrightarrow \alpha}$. Hence, for any $m$, $\omega_k \stackrel{m}{\scriptstyle{<>}}\omega'_k\Longrightarrow w \stackrel{m}{\scriptstyle{\prec\succ}}w'$. If $\omega_k=\omega'_k\eqqcolon\omega$, then both $w$ and $w'$ belong to $\omega\mG_k^{\leftrightarrow\omega}$, so $w$ is before $w'$ in $\mG_\ell$ iff $[w]_k\stackrel{\omega}{\scriptstyle{\prec\succ}}[w']_k$, and the first point is proved.

For the second point, apply Theorem \ref{ExpbAryGray} with $\ell'=k-j$ to get $\displaystyle\mG_\ell=\sum_{g\in\mG_{k-j}}g\mG_{j+1}^{\leftrightarrow\Sigma(g)}$. By definition of $j$, both $w$ and $w'$ belong to $g\mG_{j+1}^{\leftrightarrow\Sigma(g)}$ with $g=[w]^{j+1}$, hence the first equivalence. The second equivalence is then a consequence of an application of the first point of the proposition, using the fact that $\omega_{j}\neq\omega'_{j}$.

The third point is an immediate application of Theorem \ref{ExpbAryGray} to $\ell'=\ell-i$.

\end{proof}

Eventually, we have the following

\begin{theorem}\label{OrdreGrayProdGray} Let $\mW$ and $\mW'$ be two lists increasing for the Gray order, where $\mW'$ contains at least two elements. The Gray product $\mW\odot\mW'$ is increasing for the Gray order iff $\mW$ satisfies the index parity property.
\end{theorem}

\begin{proof} Write $\mW\eqqcolon(w_n)_{n\in\llbracket 0,p\llbracket}$ and $\ell'\coloneqq\|\mW'\|$, so we have $\mW\odot\mW'=\sum_{n\in\llbracket 0,p\llbracket}w_n\mW'^{\leftrightarrow n}$. Let $w$ and $w'$ be two elements of $\mW\odot\mW'$, where $w$ comes strictly before $w'$. If they belong to differents terms of the previous sum, say $w_n\mW'^{\leftrightarrow n}$ and $w_{n'}\mW'^{\leftrightarrow n'}$ with $n<n'$, then we have $[w]^{\ell'}=w_n$ and $[w']^{\ell'}=w_{n'}$. Since $\mW$ is increasing for the Gray order, we also have $w_n\prec w_{n'}$, so the third characterization given in Proposition \ref{CharacterizationsGrayOrder} implies $w\prec w'$.

Assume now (a possibility that cannot be avoided since $\mW'$ is not reduced to a single element) that both $w$ and $w'$ belong to the same $w_n\mW'^{\leftrightarrow n}$, that is: $[w]^{\ell'}=[w']^{\ell'}$. The list $\mW'$ is increasing for the Gray order, so $[w]_{\ell'}\prec[w']_{\ell'}$ iff $n$ is even, that is: $[w]_{\ell'}\stackrel{n}{\scriptstyle{\prec\succ}}[w']_{\ell'}$. Since $w_n=[w]^{\ell'}$, this is equivalent to $[w]_{\ell'}\stackrel{\Sigma([w]^{\ell'})}{\scriptstyle{\prec\succ}}[w']_{\ell'}$ (from which, again by the third characterization of Proposition \ref{CharacterizationsGrayOrder}, $w\prec w'$ follows) iff $\mW$ satisfies the index parity property, so we are done.
\end{proof}

\subsection{The $\mathcal{Z}$-Gray product}\label{ZGrayProduct}

To adapt Definition \ref{ZProduct} to the context of the Gray product, it is convenient to modify the notation a little bit: here we write $\Godot\mathcal{W}$ for what was written $\underline{\boldsymbol{\cdot}}\mW$ there.

\begin{definition}\label{QuotientFacteur} With the previous notation and the other notations of Definition \ref{ZProduct}, the {\em $\mathcal{Z}$-Gray product} of two finite lists $\mathcal{W}$ and $\mathcal{W}'$ in $\mA^*_\sim$ is defined as:
\[\mathcal{W}\underline{\odot}\mathcal{W}'\coloneqq \Godot\left(\mathcal{W}\odot\mathcal{W}'\right).\]
Whenever necessary, $\mZ$ will be regarded as a list, increasing for the Gray order.
\end{definition}

An important example of sequence defined by a $\mZ$-Gray product is the $k$-bonacci Gray code defined in Section \ref{IntroOperations}:

\begin{proposition}\label{ZFibo} Let $\mF_0\coloneqq\mM_0\coloneqq\{\varepsilon\}$ and, for any $\ell\in\N$,
\[\mM_{\ell+1}\coloneqq\sum_{i\in\llbracket 0,\min(k,\ell+1)\llbracket}10^i\overleftarrow{\mM_{\ell-i}}\qquad\text{and}\qquad\mF_{\ell+1}\coloneqq 0\mF_\ell+\mM_{\ell+1}.\]
For $\mZ=\{10^k\}$ we have $\mF_{\ell+1}=\{0,1\}\Godot\mF_{\ell}$ for any $\ell\in\N$.
\end{proposition}

\begin{proof} Define $(\mW_\ell)_{\ell\in\N}$ as $\mW_0=\{\varepsilon\}$ and $\mW_{\ell+1}\coloneqq\{0,1\}\Godot\mW_{\ell}$ for any $\ell\in\N$, and let us prove that $\mF_\ell=\mW_\ell$ for all $\ell\in\N$. First, we have $\mF_\ell=\mG^{(2)}_\ell=\mW_\ell$ for any $\ell\in\llbracket 0,k\rrbracket$. Now assume $\mF_\ell=\mW_\ell$ for some $\ell\geqslant k$. For $i\in\llbracket 0,k\llbracket$, let $\mathcal{Q}_\ell^{(i)}$ be the list such that $\mathcal{M}_{\ell-i}=1\mathcal{Q}_\ell^{(i)}$. Therefore, we can write
\[\mathcal{F}_\ell=0^k\mathcal{F}_{\ell-k}+\sum_{i\in\rrbracket k,0\rrbracket}0^i1\mathcal{Q}_\ell^{(i)}.\]
Hence, we have
\[\{0,1\}\odot\mathcal{F}_\ell=0\mathcal{F}_\ell+1\overleftarrow{\mF_{\ell}}=0\mathcal{F}_\ell+1\sum_{i\in\llbracket 0,k\llbracket}0^i1\overleftarrow{\mathcal{Q}_\ell^{(i)}}+10^k\overleftarrow{\mathcal{F}_{\ell-k}},\]
\noindent so
\begin{eqnarray*}
\{0,1\}\underline{\odot}\mathcal{F}_\ell&=&0\mathcal{F}_\ell+\sum_{i\in\llbracket 0,k\llbracket}10^i1\overleftarrow{\mathcal{Q}_\ell^{(i)}}\\
&=&0\mathcal{F}_\ell+\sum_{i\in\llbracket 0,k\llbracket}10^i\overleftarrow{\mathcal{M}_{\ell-i}}\\
&=&0\mathcal{F}_\ell+\mathcal{M}_{\ell+1}=\mathcal{F}_{\ell+1}.
\end{eqnarray*}
\end{proof}

The definition of {\em $\mZ$-geometric sequence} in Section \ref{ZGeometricSequences} will generalize the previous result.

Since $\Godot(\mathcal{W}+\mathcal{W}')=\Godot\mathcal{W}+\Godot\mathcal{W}'$ it is routinely proved that:

\begin{proposition}\label{FGrayAsso} For a given $\mathcal{Z}$-Gray product and any finite lists $\mathcal{W}$, $\mathcal{W}'$ and $\mathcal{W}''$ we have
\[\mathcal{W}\underline{\odot}(\mathcal{W}'\odot\mathcal{W}'')=\mathcal{W}\underline{\odot}(\mathcal{W}'\Godot\mathcal{W}''),\]
\[(\mathcal{W}\odot\mathcal{W}')\underline{\odot}\mathcal{W}''= \mathcal{W}\underline{\odot}(\mathcal{W}'\odot\mathcal{W}''),\]
\[(\mathcal{W}+\mathcal{W}')\underline{\odot} \mathcal{W}''=\mathcal{W}\underline{\odot}\mathcal{W}'' + \mathcal{W}'\underline{\odot}\mathcal{W}''^{\leftrightarrow\#(\mathcal{W})},\]
\[\overleftarrow{\mW}\Godot\mW'=\overleftarrow{\mW\Godot\mW'^{\leftrightarrow\#(\mW)}}.\]
In particular, if $\mW'$ is {\em palindromic} (i.e. $\mW'=\overleftarrow{\mW'}$), then $\overleftarrow{\mW}\Godot\mW'=\overleftarrow{\mW\Godot\mW'}$.

\end{proposition}

It is also easy to check that $(\mathcal{W}\underline{\odot}\mathcal{W}')\underline{\odot} \mathcal{W}''$ and $\mathcal{W}\underline{\odot}(\mathcal{W}'\underline{\odot} \mathcal{W}'')$ are equal as sets. Nevertheless, the $\mathcal{Z}$-Gray product is not associative in general. Illustrating this fact offers the occasion to present a sometimes handy way to write calculations, that consists in writing $(w_0+\cdots +w_p)$ for the list of words $(w_n)_{n\in\llbracket 0,p\rrbracket}$. The shortest example of non-associativity is provided by $\mathcal{Z}=\{11\}$, $\mathcal{W}=(1)$ and $\mathcal{W}'=\mathcal{W}''=(1+0)$. Indeed, we have
\[(\mathcal{W}\underline{\odot}\mathcal{W}')\underline{\odot}\mathcal{W}''=\Big((1)\underline{\odot}(1+0)\Big)\underline{\odot}(1+0)=(10)\underline{\odot}(1+0)=(101+100),\]
\noindent whereas:
\[\mathcal{W}\underline{\odot}(\mathcal{W}'\underline{\odot}\mathcal{W}'')=(1)\underline{\odot}\Big((1+0)\underline{\odot}(1+0)\Big)=(1)\underline{\odot}(10+00+01)=(100+101).\]

The following definition leads to Proposition \ref{ZGrayProdUnPeuAssociatif} which provides a way to overcome the non-associativity of the $\mZ$-Gray product.

\begin{definition}\label{ZDec} The {\em $\mZ$-decomposition} of the finite list $\mW$ is the sequence $(\mV_j)_{j\in\llbracket 0,2m\rrbracket}$ with minimal $m\in\N$ such that
\[\mathcal{W}=\sum_{j\in\llbracket 0,2m\rrbracket}\mV_j,\]
\noindent where, for any $j\in\llbracket 0,m\rrbracket$, $\mV_{2j}\subset\Fact_\mA(\mZ)$ and $\mV_{2j+1}\cap \Fact_\mA(\mZ)=\varnothing$.
\end{definition}

In particular we have, with the previous notation:
\[\Godot\mathcal{W}=\sum_{j\in\llbracket 0,m\llbracket}\mV_{2j+1}.\]



\begin{proposition}\label{ZGrayProdUnPeuAssociatif} The set $\mZ$ being given, let $\mW$, $\mW'$ and $\mW''$ be three finite lists. A sufficient condition for the equality
\[(\mathcal{W}\underline{\odot}\mathcal{W}')\underline{\odot}\mathcal{W}''=\mathcal{W}\underline{\odot}(\mathcal{W}'\underline{\odot}\mathcal{W}'')\]
\noindent is that, $(\mV_j)_{j\in\llbracket 0,2m\rrbracket}$ being the $\mZ$-decomposition of $\mW\odot\mW'$, for any $j\in\llbracket 0,m\llbracket$ we have $\#(\mV_{2j})\in 2\Z$ (but not necessarily for $j=m$).
\end{proposition}

\begin{remark} Another sufficient condition of the associativity of the $\mZ$-Gray product is $\mW''$ palindromic. 
\end{remark}

\begin{proof} For any $j\in\llbracket 0,2m\llbracket$, put $r_j\coloneqq\sum_{i\in\llbracket 0,j\llbracket}\#(\mV_i)$. Thus, by induction and Proposition \ref{FGrayAsso}:
\begin{eqnarray*}
\mX\coloneqq\mathcal{W}\underline{\odot}(\mathcal{W}'{\Godot}\mathcal{W}'')&=&(\mathcal{W}\odot\mathcal{W}')\underline{\odot}\mathcal{W}''\\
&=&\left(\sum_{j\in\llbracket 0,2m\rrbracket}\mV_j\right)\underline{\odot}\mathcal{W}''\\
&=&\sum_{j\in\llbracket 0,2m\rrbracket}\mV_j\Godot\mW''^{\leftrightarrow r_j}\\
&=&\sum_{j\in\llbracket 0,m\llbracket}\mV_{2j+1}\Godot\mW''^{\leftrightarrow r_{2j+1}},
\end{eqnarray*}
\noindent the last equality coming from the fact that the elements of $\mV_{2j}$ have an element of $\mZ$ as a factor for any $j\in\llbracket 0,m\rrbracket$.

Moreover, we have:
\begin{eqnarray*}
\mX'\coloneqq(\mW\Godot\mW')\Godot\mW''&=&\left(\sum_{j\in\llbracket 0,m\llbracket}\mV_{2j+1}\right)\Godot\mW''\\
&=&\sum_{j\in\llbracket 0,m\llbracket}\mV_{2j+1}\Godot\mW''^{\leftrightarrow r'_{2j+1}},
\end{eqnarray*}
\noindent where $r'_{2j+1}\coloneqq\sum_{i\in\llbracket 0,j\llbracket}\#(\mV_{2i+1})$.

If $\#(\mV_{2j})\in 2\Z$ for any $j\in\llbracket 0,m\llbracket$, then $r_{2j+1}$ and $r'_{2j+1}$ have the same parity, so $\mW''^{\leftrightarrow r_{2j+1}}=\mW''^{\leftrightarrow r'_{2j+1}}$, therefore $\mX=\mX'$.\end{proof}

Observe that, in the previous proof, if $\mW''$ is palindromic, then both $\mW''^{\leftrightarrow r_{2j+1}}$ and $\mW''^{\leftrightarrow r'_{2j+1}}$ are equal to $\mW''$ for any $j\in\llbracket 0,m\llbracket$, so we also have $\mX=\mX'$, regardless of the $\mZ$-decomposition of $\mW\odot\mW'$.

It is tempting to ask for a necessary and sufficient condition for the associativity of the $\mZ$-Gray product. Even if we will not use it in the following, here is an answer, that shows that Proposition \ref{ZGrayProdUnPeuAssociatif} is quite close.

\begin{proposition}\label{ZGrayProdUnPeuAssociatifBis} With the notation of Proposition \ref{ZGrayProdUnPeuAssociatif} and its proof, the equality 
\[(\mathcal{W}\underline{\odot}\mathcal{W}')\underline{\odot}\mathcal{W}''=\mathcal{W}\underline{\odot}(\mathcal{W}'\underline{\odot}\mathcal{W}'')\]
\noindent holds iff, $(\mV_j)_{j\in\llbracket 0,2m\rrbracket}$ being the $\mZ$-decomposition of $\mW\odot\mW'$, for any $j\in\llbracket 0,m\llbracket$ (but not necessarily for $j=m$) we have $\mV_{2j+1}\Godot\mW''=\mV_{2j+1}\Godot\overleftarrow{\mW''}$ for any $j$ for which $r_{2j+1}$ and $r'_{2j+1}$ are not of the same parity.
\end{proposition}

\begin{proof} By the proof of Proposition \ref{ZGrayProdUnPeuAssociatif} we only have to investigate the case for which there exists a $q\in\llbracket 0,m\llbracket$ such that $r_{2q+1}$ and $r'_{2q+1}$ are not of the same parity. Consider the smallest such $q$, so we have
\[\sum_{j\in\llbracket 0,q\llbracket}\mV_{2j+1}\Godot\mW''^{\leftrightarrow r_{2j+1}}=\sum_{j\in\llbracket 0,q\llbracket}\mV_{2j+1}\Godot\mW''^{\leftrightarrow r'_{2j+1}}\eqqcolon\mY.\]
By construction, $\mY+\mV_{2q+1}\Godot\mW''^{\leftrightarrow r_{2q+1}}$ is the beginning of the list $\mX$, and $\mY+\mV_{2q+1}\Godot\mW''^{\leftrightarrow r'_{2q+1}}$ is the beginning of the list $\mX'$. Without the assumption $\mV_{2q+1}\Godot\mW''=\mV_{2q+1}\Godot\overleftarrow{\mW''}$, these two parts coming right after $\mY$ in $\mX$ and $\mX'$ would be equal as sets but different as lists, forcing $\mX$ and $\mX'$ to be different. The end of the proof is a simple induction on $q$.
\end{proof}

\begin{proposition}\label{IndexParityPreservation} Le $\mW$ and $\mW'$ be two finite lists in $\N^*_\sim$ that both satisfy the index parity property. For a given $\mZ$, let $(\mV_i)_{i\in\llbracket 0,2m\rrbracket}$ be the $\mZ$-decomposition of $\mW\Godot\mW'$. Then, $\mW\Godot\mW'$ satisfies the index parity property iff $m=0$ or $\#(\mV_i)\in 2\N$ for any $i\in\llbracket 0,2(m-1)\rrbracket$.
\end{proposition}

\begin{proof} Proposition \ref{IndexParityPreservationOdot} gives that $\mW\odot\mW'$ satisfies the index parity property. For a given list $\mX$ satisfying the index parity property, removing elements from it provides a list that still satisfies the same property iff the removed elements go by pair of consecutive elements and/or are at the end of $\mX$, hence the result.
\end{proof}

\subsection{$\mZ$-geometric sequences and power-associativity}\label{ZGeometricSequences}

In the following, an integer $b\geqslant 2$ is given, and $\mA=\llbracket 0,b\llbracket$.

\begin{definition}\label{Geom} Let $\mathcal{Z}$ be a subset of $\mA^*$ and $\mathcal{W}$ a finite list in $\mA^*$. The {\em $\mZ$-geometric sequence with ratio $\mW$} is defined by $\mW_0=\{\varepsilon\}$ and $\mW_{\ell+1}=\mW\Godot\mW_{\ell}$ for any $\ell\in\N$. 
\end{definition}

In all the present section we stick to the notations of Definition \ref{Geom}.

\begin{proposition}\label{ZGeomPrefPart} If $\mW$ is prefix-partitioned and does not contain twice the same element, then $\mW_\ell$ is prefix-partitioned for any $\ell\in\N$.
\end{proposition}

\begin{proof} Assume the property is satisfied for some $\mW_\ell$ with $\ell\in\N^*$. Since the property is stable under deletion, it is enough to show that $\mW\odot\mW_{\ell}$ is prefix-partitioned. Let $\tilde{w}\in\mA^*_\sim$ and $i\in\N$. For any list $\mV$ we write $[\mV]^i$ for the list $\{[v]^i\ :\ v\in\mV\}$. Assume $\tilde{w}\in[\mW\odot\mW_\ell]^i$. We therefore either have $\tilde{w}\in[\mW]^j$ for some $j\leqslant\|\mW\|$ (in the case $i\geqslant\|\mW_\ell\|$) or that there exists $w\in\mW$ and $w'\in\mA^{*}$ such that $\tilde{w}=ww'$ (in the case $i<\|\mW_\ell\|$).

In the first case, since, by hypothesis, $\mW$ is prefix-partitioned, the sublist $\mI\coloneqq\{v\in\mW\ :\ [v]^i=\tilde{w}\}$ is an interval of $\mW$, so we have $\mW=\mW'+\mI+\mW''$ for some lists $\mW'$ and $\mW''$ that do not contain any $v\in\mA^*$ such that $[v]^i=\tilde{w}$. Hence, we have
\begin{eqnarray*}
\mW\odot\mW_\ell&=&(\mW'+\mI+\mW'')\odot\mW_\ell\\
&=&\mW'\odot\mW_\ell+\mI\odot\mW_\ell^{\leftrightarrow\#(\mW')}+\mW''\odot\mW_\ell^{\leftrightarrow\#(\mW'+\mI)},
\end{eqnarray*}
\noindent so $\{v\in\mW\odot\mW_\ell\ :\ [v]^i=\tilde{w}\}=\mI\odot\mW_\ell^{\leftrightarrow\#(\mW')}$, which is an interval of $\mW\odot\mW_\ell$.

Now consider the second case. Since $\mW$ is prefix-partitioned without containing twice the same element we can write $\mW\eqqcolon\mX+\{w\}+\mY$ for some lists $\mX$ and $\mY$ not containing $w$. We can then write
\[\mW\odot\mW_\ell=(\mX+\{w\}+\mY)\odot\mW_\ell=\mX\odot\mW_\ell+w\mW_\ell^{\leftrightarrow\#(\mX)} +\mY\odot\mW_\ell^{\leftrightarrow\#(\mX)+1}.\]

By the properties of $\mX$ and $\mY$, the sublist of $\mW\odot\mW_\ell$ made of words $v$ such that $[v]^i=\tilde{w}$ is therefore contained in the interval $w\mW_\ell^{\leftrightarrow\#(\mX)}$. By the induction hypothesis, we also have $\mW_\ell=\mX'+w'\mW'+\mY'$ for some lists $\mX'$, $\mW'$ and $\mY'$, where neither $\mX'$ nor $\mY'$ contain any $v$ such that $[v]^{\|\mW_\ell\|-|w'|}=w'$. Therefore, we have
\[w\mW_\ell^{\leftrightarrow\#(\mX)}=w(\mX'+w'\mW'+\mY')^{\leftrightarrow\#(\mX)}=(w\mX'+\tilde{w}\mW'+w\mY')^{\leftrightarrow\#(\mX)},\]
\noindent so $\{v\in w\mW_\ell\ :\ [v]^i=\tilde{w}\}=\tilde{w}\mW'^{\leftrightarrow\#(\mX)}$, so it is an interval.
\end{proof}

\begin{definition}\label{GeomPower} The $\mZ$-geometric sequence with ratio $\mW$ is {\em $\mZ$-power-associative} if, for any $\ell$, $\ell'\in\N$, we have $\mW_{\ell+\ell'}=\mW_\ell\Godot\mW_{\ell'}$. In this case, we can write $\mW_\ell\eqqcolon\mW^{\Godot\ell}$ for $\ell\in\N$.
\end{definition}

For $\mZ$ empty we are back to the case of the standard $b$-ary Gray code of Corollary \ref{Exponentiation}, but since the $\mZ$-Gray product is not associative, not all $\mZ$-geometric sequence are $\mZ$-power-associative. A minimal counterexample is given by $\mZ=\{11\}$ and $\mW=\{0,1\}$, for which we have $\mW_3=\{000,001,010,101,100\}$ but $\mW_2\Godot\mW=\{000,001,010,100,101\}$. This can happen even  under what could appear to be a quite natural assumption on $\mZ$ bearing in mind Proposition \ref{ZGrayProdUnPeuAssociatif}. A minimal example of this is the $\mZ$-geometric sequence with ratio $\mW=\{0,1\}$, for $\mZ=\{111,101\}$: even if $\mW_1$, $\mW_2$ and $\mW_3$ do satisfy the assumption of Proposition \ref{ZGrayProdUnPeuAssociatif}, the sequence is not $\mZ$-power-associative since the four last elements of $\mW_5$ are $\{10010,10011,10001,10000\}$ whereas those of $\mW_4\Godot\mW$ are $\{10011,10010,10000,10001\}$.

Here is a quite general result.




\begin{theorem}\label{ClassExp} Let $b\geqslant 2$ be an integer, let $\mZ\subset\mG_k$ for some $k\in\N^*$, satisfying either that $\mZ=\rrbracket z,\max(\mG_k)\rrbracket_{\mG_k}$ for some $z\in\mG_k$ or that, for some list $\mY\subset\mG_k\cap\Sigma^{\text{-}1}(0)$, we have $\mZ=\mY\ \!\cup\ \!\suc_{\mG_k}(\mY)$, with the complementary assumption in this case that, for any $w\in\mA^k\backslash\mY$ and any $\alpha\in\llbracket 0,b/2\llbracket$, we have $[\alpha w]^1\in\mZ\Longleftrightarrow [\alpha\ \!\suc_{\mG_k}(w)]^1\in\mZ$. The $\mZ$-geometric sequence with ratio $\mA$ is $\mZ$-power-associative.
\end{theorem}

Before proving this theorem, let us make a few observations.

It is not possible to unify directly the two cases in a single statement, as the case $b=2$ and $\mZ=\{1100,1101,1000\}$ shows: in this case, $\mW_6$ and $\mW_5\Godot\mW$ are different from their terms indexed by $21$ onwards, so the sequence $(\mW_\ell)_{\ell\in\N}$ is not $\mZ$-power-associative.

The case $\mZ=\rrbracket z,\max(\mG_k)\rrbracket_{\mG_k}$ contains the $(k+1)$-bonacci case (with $b=2$ and $\mZ=\{10^{k-1}\}$), together with many others given by taking $\#(\mZ)>1$ and/or $b>2$. For example, the case $b=3$ and $z=20$ shows that the $\mZ$-geometric sequence of ternary words not in $\Fact_{\{0,1,2\}}(\{22,21\})$ (whose combinatorics corresponds to the linear recurring sequence $(u_\ell)_{\ell\in\N}$ defined by $u_0=1$, $u_1=2$ and $u_{\ell+2}=2u_{\ell+1}+u_\ell$ for $\ell\in\N$) is power-associative.

When $b$ is even, to recover power-associativity when $\mZ$ is not of one of the two forms considered in Therem \ref{ClassExp}, we may replace $\mZ$ by $\mZ'\coloneqq\mZ\odot\mA$, which leads to a $\mZ'$-geometric sequence whose properties remain quite close to the ones of the $\mZ$-geometric sequence.

Eventually, the assumption that $\mW=\mA$ is not that restrictive since, for a finite list $\mW$ containing $\varepsilon$, strictly increasing for the Gray order and satisfying the index parity property, we can consider its increasing and one-to-one recodage to $\llbracket 0,\#(\mW)\llbracket$ and define a new form of $\mZ$ accordingly. Only the case $\varepsilon\notin\mW$ deserve specific consideration, briefly done after the following proof.

\begin{proof}[Proof of Theorem \ref{ClassExp}.] Theorem \ref{OrdreGrayProdGray} gives by induction that, for any $\ell\in\N$, $\mA_\ell$ is increasing for the Gray order. Hence, since $\mA_{\ell+\ell'}$ and $\mA_\ell\Godot\mA_{\ell'}$ are equal as sets for any $\ell$, $\ell'\in\N$, we only have to show that $\mA_\ell\Godot\mA_{\ell'}$ too is increasing for the Gray order to obtain the equality of the two lists. Again by Theorem \ref{OrdreGrayProdGray} it is therefore sufficient to prove that $\mA_\ell$ satisfies the index parity property for any $\ell\in\N$. For any $\ell< k$ we have $\mA_\ell=\mG_\ell$, so this is true in this case. The hypotheses on $\mZ$ implies that this is true as well for $\ell=k$.

In the case $\mZ=\rrbracket z,\max(\mG_k)\rrbracket_{\mG_k}$, then, by induction, $\mA_\ell$ is an interval of $\mG_\ell$ containing $\varepsilon$, hence satisfies the index parity property, so we are done in this case.

Consider the second form for $\mZ$. Considering $\mA_k$ as the complement of $\mZ$ in $\mG_k$ gives that, for $\ell=k$, there exists $\mH_\ell\subset\mG_\ell$ such that 
\[\mA_\ell=\sum_{w\in\mH_\ell}\left\{w,\suc_{\mG_k}(w)\right\}\]
\noindent where $\Sigma(\mH_\ell)=0$. By induction, let assume that such a decomposition exists for some $\ell\geqslant k$, and let us prove that it exists also for $\mA_{\ell+1}$, so that the theorem is proved. Put $\mP_w\coloneqq\{w,\suc_{\mG_k}(w)\})$ for all $w\in\mH_\ell$ and $\beta\coloneqq\lfloor b/2\rfloor$. Assuming $(2\beta)\mX$ is the empty list when $b$ is even (whatever the list $\mX$ is), we can write:

\begin{eqnarray*}
\lefteqn{\mA\odot\mA_\ell}\\
&=&\sum_{\alpha\in\mA}\alpha\mA_\ell^{\leftrightarrow \alpha}\\
&=&\sum_{\alpha'\in\llbracket 0,\beta\llbracket}\left((2\alpha')\mA_\ell+(2\alpha'+1)\overleftarrow{\mA_\ell}\right)+(2\beta)\mA_\ell\\
&=&\sum_{\alpha'\in\llbracket 0,\beta\llbracket}\left((2\alpha')\left(\sum_{w\in\mH_\ell}\mP_w\right)+(2\alpha'+1)\overleftarrow{\sum_{w\in\mH_\ell}\mP_w}\right)+(2\beta)\mA_\ell \\
&=&\sum_{\alpha'\in\llbracket 0,\beta\llbracket}\left(\sum_{w\in\mH_\ell}(2\alpha')\mP_w+\sum_{w\in\overleftarrow{\mH_\ell}}(2\alpha'+1)\overleftarrow{\mP_w}\right)+\sum_{w\in\mH_\ell}(2\beta)\mP_w.
\end{eqnarray*}

Proposition \ref{CharacterizationsGrayOrder} gives that, in this expression, the terms $(2\alpha')\mP_w$, $(2\alpha'+1)\overleftarrow{\mP_w}$ and $(2\beta)\mP_w$ are either empty or of the form $\left\{h,\suc_{\mG_{k+1}}(h)\right\}$ with $\Sigma(h)=0$. Now, to get $\mA_{\ell+1}$, we have to compute $\mA\Godot\mA_\ell$, i.e. to remove, in the previous expression, all the words in $\Fact_\mA(\mZ)$. To ensure that this does not break the properties of our decomposition, it is sufficient to show that, when proceeding to the removal, each term $(2\alpha')\mP_w$, $(2\alpha'+1)\overleftarrow{\mP_w}$ and $(2\beta)\mP_w$ either remains in full or is fully cancelled (i.e. without one of its elements remaining and the other one being removed).

Consider for example the case of $(2\alpha')\mP_w\eqqcolon\{(2\alpha')w,(2\alpha')w'\}$ (the others being similar), and write $\widetilde{w}\coloneqq[w]^{\ell-(k-1)}$ and $\widetilde{w'}\coloneqq[w']^{\ell-(k-1)}$. Since $\mP_w\subset\mA_\ell$, the word $(2\alpha')w$ is to be removed iff $(2\alpha')\widetilde{w}\in\mZ$. Similarly, $(2\alpha')w'$ is to be removed iff $(2\alpha')\widetilde{w'}\in\mZ$. Of course, if $\widetilde{w}=\widetilde{w'}$, then either both elements of $(2\alpha')\mP_w$ thus belong to $\mA_{\ell+1}$ or none of them, so we are done in this case. If $\widetilde{w}\neq\widetilde{w'}$, then, by Proposition \ref{CharacterizationsGrayOrder}, $\widetilde{w'}=\suc_{\mG_k}(\widetilde{w})$, so the conclusion is given by the complementary assumption made on $\mZ$.
\end{proof}

When $\mZ$ is an interval of $\mG_k$ containing $\varepsilon$, the corresponding $\mZ$-geometric sequence with ratio $\mA$ is not increasing in the sense of inclusion, so it is natural in this case to sum the terms to get a meaningful Gray code. We do not try to provide a general result here, only state a result in the $k$-bonacci case that can be regarded as a synthesis of the two different definitions of the Fibonacci (or $k$-bonacci) Gray code as defined in Section \ref{IntroOperations} as well as Proposition \ref{ZFibo}:


\begin{theorem}\label{kBonacciAvec0k} Put $b=2$, so $\mA=\{0,1\}$. Let $\mZ\coloneqq\{\varepsilon\}\subset\mG_k$ and $\mZ'\coloneqq\{10^k\}\subset\mG_{k+1}$, defining the $\mZ$-geometric sequence $(\mA_\ell)_{\ell\in\N}$ and the $\mZ'$-geometric sequence $(\mA'_\ell)_{\ell\in\N}$, both with ratio $\mA$. For any $\ell\geqslant k$ we have
\[\mA'_\ell=\mA'_{\ell-k}+ \mA_{\ell}.\]
\end{theorem}

\begin{proof} By Theorem \ref{OrdreGrayProdGray}, both $\mA_\ell$, $\mA'_{\ell-k}$ and $\mA'_\ell$ are increasing for the Gray order. Moreover, for $w\in\mA_\ell$ and $w'\in\mA'_{\ell-k}$ we have $\|w\|>\ell-k$ (by definition of $\mZ$) and $\|w'\|\leqslant\ell-k$, so $w\prec w'$. Therefore, it only remains to show that $\mA'_\ell$ and $\mA'_{\ell-k}+\mA_\ell$ are equal as sets. Of course we have that, as sets, $\mA'_{\ell-k}$ and $\mA_\ell$ are parts of $\mA'_\ell$. Now, let $w\in\mA'_\ell$. If $\|w\|\leqslant \ell-k$ then $w\in\mA'_{\ell-k}$,  so assume $\|w\|>\ell-k$. Hence, there exists $k'\in\llbracket 0,k\llbracket$ and $w'$ such that $w=0^{k'}1w'$. Moreover, we must have $w'\notin\Fact_{\{0,1\}}(\{0^k\})$, otherwise $10^k$ would be a factor of $w$. Hence, $w'\in\mA_\ell$, and we are done.
\end{proof}


\subsection{The flipping digit property}\label{TheFlippingDigitsProperty}

We still use the notations of the previous section. By Proposition \ref{ZGeomPrefPart}, whatever $\mZ$ is, for any word $\tilde{w}$ and any $\ell\in\N$, the sublist $\tilde{w}\mA_\ell(\tilde{w})^{\leftrightarrow\Sigma(\tilde{w})}\coloneqq\{w\in\mA_\ell\ :\ [w]^{\ell-|\tilde{w}|}=\tilde{w}\}$ is an interval of $\mA_\ell$, so there exists a list $\mH_\ell\subset\mG_k$ such that $h\mA_\ell(h)^{\leftrightarrow\Sigma(h)}$ is nonempty for any $h\in\mH_\ell$ and such that $\mA_\ell=\sum_{h\in\mH_\ell}h\mA_\ell(h)^{\leftrightarrow\Sigma(h)}$. Note that, for $h\in\mG_{k+1}$, $\mA_\ell(h)=\varnothing$ does not implies $h\in\mZ$, as the example $\mA=\llbracket 0,1\llbracket$, $\mZ=\{110,111\}$ and $\mA_4(011)=\varnothing$ easily shows.

With these notations we can state the following

\begin{theorem}\label{FlidGen} Let $\mK_\ell$ be the sublist of $\mA\otimes\mH_\ell$ containing the pairs $(\alpha,h)$ such that $\alpha h\notin\Fact_\mA(\mZ)$. For a given $\ell\in\N$, let $\mI\subset\N^*$ be a set containing $H(\mA_\ell)$ and such that, for any consecutive elements $(\alpha,h)$ and $(\alpha',h')$ of $\mK_\ell$ we have
\begin{equation}\label{HammingStable}\tag{\textasteriskcentered}
H\Big(\Last\big(\alpha h\mA_\ell(h)^{\leftrightarrow \Sigma(\alpha h)}\big),\Init\big(\alpha'h'\mA_\ell(h')^{\leftrightarrow\Sigma(\alpha' h')}\big)\Big)\in\mI.
\end{equation}
We then have $H(\mA_{\ell+1})\subset\mI$.

 





\end{theorem}

Before proving this result, it should be remarked that numerous cases, it is quite easy to check that the property (\ref{HammingStable}) is satisfied, with the help of the following observations:

\begin{itemize}

\item $(\mH_\ell)_{\ell\in\N}$ is decreasing for inclusion, so $(\mK_\ell)_{\ell\in\N}$ is ultimately constant;

\item when $\alpha=\alpha'$ and $h$ and $h'$ are successive elements of $\mH_\ell$ (their order depending on the parity of $\alpha$), the relation (\ref{HammingStable}) is essentially a reformulation of the relation $\mA_\ell=\sum_{h\in\mH_\ell}h\mA_\ell(h)^{\leftrightarrow\Sigma(h)}$ and $H(\mA_\ell)\subset\mI$;

\item when $h=h'$, the relation (\ref{HammingStable}) is essentially equivalent to the assumption that $\alpha$ and $\alpha'$ are not of the same parity;

\item in several cases it is possible to know explicitely what the extremal elements of $\alpha h\mA_\ell(h)^{\leftrightarrow \Sigma(\alpha h)}$ are, to prove that (\ref{HammingStable}). It is the case, for example, when the words of $\mZ$ does not contain neither the letter $0$ nor the letter $b-1$.
\end{itemize}

As an exercise, let us show how to use Theorem \ref{FlidGen} to prove that the the $k$-bonacci Gray code satisfies the $\{1\}$-flipping digit property. We have $\mZ=\{10^k\}$, and it is immediate to prove that $\mH_\ell=\mH\coloneqq\mG_{k+1}\backslash\{10^k\}$ for any $\ell\geqslant k$ (since for any $h\in\mG_{k+1}\backslash\{10^k\}$ and any $\ell\geqslant k$ we have, for example, $h1^{\ell-(k+1)}\in\mA_\ell$). We also easily have, for any $\ell\geqslant k$:
\[\mK_{\ell}=\mK\coloneqq\sum_{h\in\mG_{k+1}\backslash\{10^k\}}(0,h)+\sum_{h\in\overleftarrow{\mG_{k+1}\backslash\{\varepsilon,10^k\}}}(1,h).\]

We also have $H(\mA_{k+1})=\{1\}$, so assume by induction that $H(\mA_\ell)=\{1\}$ for some $\ell>k$. Two consecutive elements $(\alpha,h)$ and $(\alpha',h')$ of $\mK$ satisfy one of the following conditions, each leading the the relation (\ref{HammingStable}):

\begin{itemize}

\item $\alpha=\alpha'=0$ and $h'=\suc_{\mG_{k+1}}(h)$. In this case, we have to prove that $H\left(\Last(h\mA_\ell(h)^{\leftrightarrow\Sigma(h)}),\Init(h'\mA_\ell(h')^{\leftrightarrow\Sigma(h')})\right)=1$, which is true by the induction hypothesis on $\mA_\ell=\sum_{g\in\mH}g\mA_\ell(g)^{\leftrightarrow\Sigma(g)}$.


\item $\alpha=0$, $\alpha'=1$ and $h=h'=10^{k-1}1$. In this case, we can simply use the third previous observation about the fact that $\alpha$ and $\alpha'$ are of different parity.

\item $\alpha=\alpha'=1$ and $h'=\pre_{\mG_{k+1}}(h)$, which is similar to the first case.

\end{itemize}

\begin{remark} A simple example of $\mZ$-geometric sequences $(\mA_\ell)_{\ell\in\N}$ for which $(H(\mA_\ell))_{\ell\in\N}$ is unbounded is given by $\mA=\{0,1\}$ and $\mZ=\{110\}$. It is an easy exercise to show that, in this case, for any $\ell\in\N^*$, $H(\mA_\ell)=\llbracket 1,\ell\llbracket$.
\end{remark}

\begin{proof}[Proof of Theorem \ref{FlidGen}.] By Proposition \ref{IndexParityPreservation} and Theorem \ref{OrdreGrayProdGray}, for any $\ell\in\N$ the list $\mA_\ell$ is increasing for the Gray order.  Therefore, we have
\begin{eqnarray*}
\mA_{\ell+1}&=&\Godot\left(\sum_{\alpha\in\mA}\alpha\mA_\ell^{\leftrightarrow \Sigma(\alpha)}\right)\\
&=&\Godot\left(\sum_{\alpha\in\mA}\alpha\left(\sum_{h\in\mH_\ell}h\mA_\ell(h)^{\leftrightarrow\Sigma(h)}\right)^{\leftrightarrow \Sigma(\alpha)}\right)\\
&=&\sum_{\alpha\in\mA}\sum_{h\in\mH_\ell^{\leftrightarrow \Sigma(\alpha)}}\Godot\left(\alpha h\mA_\ell(h)^{\leftrightarrow \Sigma(\alpha h)}\right)\\
&=&\sum_{(\alpha,h)\in\mA\otimes\mH_\ell}\Godot\left(\alpha h\mA_\ell(h)^{\leftrightarrow \Sigma(\alpha h)}\right).
\end{eqnarray*}

By definition of $\mA_\ell(h)$ we have
\[\Godot\left(\alpha h\mA_\ell(h)^{\leftrightarrow \Sigma(\alpha h)}\right)=\begin{cases}\varnothing & \text{if $\alpha h\in\Fact_\mA(\mZ)$}\\ \alpha h\mA_\ell(h)^{\leftrightarrow \Sigma(\alpha h)}&\text{otherwise}\end{cases}.\]
Hence, by definition of $\mK_\ell$ we have
\[\mA_{\ell+1}=\sum_{(\alpha,h)\in\mK_\ell}\alpha h\mA_\ell(h)^{\leftrightarrow \Sigma(\alpha h)}.\]

By induction hypothesis, for any $(\alpha,h)\in\mA\otimes\mH_\ell$ we have $H(\alpha h\mA_\ell(h)^{\leftrightarrow \Sigma(\alpha h)})\subset \mI$. From (\ref{HammingStable}) we then get the result.\end{proof}

%


\end{document}